\PassOptionsToPackage{gray}{xcolor}
\documentclass[pdflatex,sn-mathphys,Numbered]{sn-jnl}

\usepackage{amsmath,amssymb,amsthm}
\usepackage{manyfoot} 
\usepackage{listings}
\usepackage{xcolor}
\usepackage{booktabs}
\usepackage{graphicx}
\usepackage{tikz}
\usepackage{seqsplit}
\usepackage{tabularx}

\hypersetup{
 colorlinks=true,
 linkcolor=black,
 citecolor=black,
 urlcolor=black,
 pdftitle={Formal verification of Romanov's triplet logic},
 pdfsubject={cs.LO, cs.AI},
 pdfauthor={Dmitry V. Alexandrov},
 pdfkeywords={formal verification, Rocq, SAT, triplet logic, sliding-window CNF, polynomial-time bounds, verified complexity}
}

\DeclareTextFontCommand{\texttt}{\ttfamily\hyphenchar\font=45\relax}

\usetikzlibrary{shapes.geometric, arrows.meta, positioning, calc, fit, backgrounds, decorations.pathreplacing}

\theoremstyle{definition}
\newtheorem{definition}{Definition}[section]
\newtheorem{theorem}{Theorem}[section]
\newtheorem{lemma}{Lemma}[section]
\newtheorem{corollary}{Corollary}[section]
\newtheorem{remark}{Remark}[section]

\tikzset{
 box/.style={draw, rounded corners, minimum width=2cm, minimum height=0.8cm, align=center, fill=black!8},
 arrow/.style={-{Stealth[length=2mm]}, thick},
 label/.style={font=\small\bfseries},
 tierbox/.style={draw, rectangle, minimum width=1.2cm, minimum height=0.6cm, align=center, fill=black!8},
 triplet/.style={draw, circle, minimum size=0.5cm, align=center, fill=black!5, font=\tiny},
}

\begin{document}

\title[Formal verification of Romanov's triplet logic]{Formal verification of Romanov's triplet logic:\\
 A verified filter for sliding-window 3-CNF\\
 with application to structured formulas}

\author*[1]{\fnm{Dmitry V.} \sur{Alexandrov}}\email{dvalexandrov@hse.ru}

\affil*[1]{\orgname{HSE University}, \orgaddress{\city{Moscow}, \country{Russia}}}

\abstract{We present the first mechanised formalisation of Romanov's Triplet Logic (TLS) in the Rocq proof assistant. TLS is a combinatorial framework originally motivated by Boolean satisfiability, based on triplet structures and a filter that we call Simple Vertex Intersection (SVI). We formalise the core of TLS, including its translation from 3-CNF, the clearing procedure, and the SVI algorithm. For the well-formed sliding-window fragment, we prove explicit polynomial-time bounds for the filter stages and verify the translation and intersection operations. Our main contribution is a precise correctness boundary: for general formulas, SVI non-emptiness is necessary but not sufficient for satisfiability; for aligned structures, we prove a full bi-implication, extended to systems of structures. We also formalise the grouped-window translation and provide a formal counterexample to its completeness. We introduce VFR (Verified Filter for Romanov's triplet logic), an extracted OCaml prototype that implements a verified decision procedure for the sliding-window fragment and a sound filter for general 3-CNF, with a Python runtime and Docker packaging. Benchmarks corroborate the predicted behaviour, and the complete toolchain is available as a curated Zenodo artifact. The Rocq development comprises over 23,000 lines of code, with 424 proved lemmas and no unproved assumptions.}

\keywords{Formal verification, interactive theorem proving, SAT solving,
 proof assistants, mechanised mathematics, Rocq, triplet logic,
 polynomial-time bounds, one-sided filter}

\maketitle

\section{Introduction}
\label{sec:intro}

The Boolean satisfiability problem (SAT) is a cornerstone of computational complexity 
theory, with applications ranging from hardware verification to automated planning~
\cite{handbook-sat}. Despite decades of research, no polynomial-time algorithm is known for 3-SAT, 
and the prevailing conjecture is that $\mathsf{P} \neq \mathsf{NP}$~
\cite{cook1971complexity,levin1973universal}. Nevertheless, numerous alternative approaches 
have been proposed, each offering new structural insights into the problem.

One such approach is Romanov's \emph{Triplet Logic} (TLS), introduced in 
``Non-Orthodox Combinatorial Models Based on Discordant Structures''~
\cite{romanov2011nonorthodox}. TLS encodes a 3-CNF (Conjunctive Normal Form) formula as a \emph{Compact Triplets 
Formula} (CTF), transforms it into a \emph{Compact Triplets Structure} (CTS) 
containing all triplets not forbidden by the corresponding clause group, and applies \emph{Simple Vertex Intersection} (SVI), which 
constructs hyperstructures via tier-wise intersection. Romanov developed this approach with the goal of efficient SAT solving via 
tier-wise triplet analysis and hyperstructure intersection.

Following Romanov's original terminology, we refer to these tiered combinatorial objects as \emph{Compact Triplets Structures}. Throughout the paper, the abbreviation CTS always denotes Romanov's construction, formalised and extended here in the Rocq proof assistant.

Romanov's key insight is a novel \emph{geometric decomposition}: instead of 
searching over variable assignments directly, TLS searches over paths through a 
layered graph of triplet tiers. This perspective is structurally distinct from 
classical DPLL (Davis--Putnam--Logemann--Loveland) / CDCL (Conflict-Driven Clause Learning) approaches and offers a new combinatorial perspective on constraint 
satisfaction problems. In this paper we treat TLS as a \emph{self-contained 
mathematical framework} and subject it to rigorous formal analysis, using 
3-CNF formulas as a motivating source of benchmark instances rather than as 
the primary object of study.

\textbf{Our contributions.} We formalise the core of TLS in the Rocq proof 
assistant~\cite{coq} (version 9.1.1), complemented by exhaustive model checking. Our findings are:

\begin{enumerate}
 \item We formalise CTF, CTS, hyperstructures, and SVI in Rocq, proving basic 
 correctness lemmas about path construction and compatibility 
 (Section~\ref{sec:formalisation}).
 
 \item We \emph{clarify the correctness boundary} of TLS: the existence of a 
 joint satisfying set (JSS) implies SVI's non-emptiness for non-empty structures (JSS $\Rightarrow$ SVI non-emptiness, proved in Rocq), 
 but the converse does \emph{not} hold in general 
 (Section~\ref{sec:boundary}). We validate the failure of the reverse direction 
 with both Rocq counterexamples and exhaustive Python-based model checking 
 (Section~\ref{sec:tla}).
 
 \item We introduce \emph{VFR} (\textbf{V}erified \textbf{F}ilter for \textbf{R}omanov's triplet logic), a prototype implementation (Section~\ref{sec:vfr}). It uses the verified clause-by-clause pipeline for well-formed sliding-window CNF and falls back to an unverified grouped-window heuristic for general 3-CNF (Section~\ref{sec:scope}).
 
 \item We benchmark VFR on random and structured 3-CNF instances (Section~\ref{sec:experiments}). On the verified fragment agreement is 100\%; on general random instances the filter is ineffective.
 
 \item We prove a new theorem in Rocq establishing that the aligned intersection 
 of two CTS has a non-empty set of full-length paths \emph{if and only 
 if} a compatible joint satisfying set exists (Theorem~\ref{thm:compat}). This 
 provides the formal foundation for the post-check. The equivalence is conceptually straightforward---it is essentially a restatement of the definition of 
 a compatible path in the intersection---but its verified mechanisation yields an 
 extracted, correct-by-construction decision procedure for aligned structures.
 
 \item We extend this bi-implication to \emph{systems} of $k$ aligned structures 
 (Theorem~\ref{thm:compat-k}), proving that the systemic tier-wise intersection 
 contains a full-length path iff a compatible joint satisfying set exists for the 
 entire system.
 
 \item We formalise the clearing procedure's termination using a tight 
 measure ($\mathrm{cts\_size}$) and prove a semantic fixed-point 
 characterisation: every surviving triplet has compatible neighbours 
 in adjacent tiers (Section~\ref{sec:formalisation}).
 
 \item We identify a \emph{semantic gap} between the weak 
 formula-level predicate $\mathrm{satisfies\_ctf}$ (which checks each 
 3-bit window independently) and structure-level path existence 
 (\texttt{build\_paths\_all} requires globally compatible consecutive 
 triplets). We formalise this in Theorem~\ref{thm:gap}: 
 there exist CTFs that admit locally consistent assignments under 
 $\mathrm{satisfies\_ctf}$ yet yield no compatible path after clearing 
 (Section~\ref{sec:boundary}). This motivates the aligned-intersection 
 approach, for which we recover completeness.
\end{enumerate}

\textbf{Implications.} Our work \emph{sharpens} the understanding of TLS. SVI is 
not a complete decision procedure. However, when SVI reports emptiness, the formula 
is guaranteed unsatisfiable---a property we prove formally.

Furthermore, TLS offers a novel \emph{combinatorial visualisation} of SAT instances 
through tiered triplet structures. This geometric perspective may aid in educational 
contexts and in analysing formula structure before invoking expensive CDCL solvers.

\textbf{Why this matters for automated reasoning.}
The paper's primary audience is the formal-verification and automated-reasoning
community rather than the SAT-competition community. Our goal is not to
outperform CDCL on benchmark suites---a task for which decades of engineering
have produced highly optimised, unverified solvers---but to demonstrate that a
non-classical combinatorial framework can be fully mechanised, its correctness
boundary exactly determined, and its polynomial fragments certified with
concrete complexity bounds. Such mechanised reconstructions are valuable
because they (i) expose hidden assumptions that informal descriptions miss,
(ii) produce certified building blocks that compose into larger verified systems,
and (iii) provide rigorous foundations for teaching and further research.

More concretely, TLS offers three affordances that complement verified CDCL
solvers. \emph{As an intermediate representation,} triplet tiers make
variable-interaction structure explicit: a formula analyst can inspect which
triplets survive clearing and immediately see local inconsistencies that would
be buried in a flat clause list. \emph{As a preprocessor,} the verified SVI
filter can be placed in front of \emph{any} solver; when it reports UNSAT the
answer is proof-carrying, and when it is inconclusive the solver falls back to
standard search with no loss. \emph{As a certification target,} the geometric
path-building algorithm yields a concrete witness---a sequence of compatible
triplets---that is easier to audit than a DRAT (Delete-Resolution-Asymmetric-Tautology) trace. These properties do not
make TLS faster than CDCL, but they make it structurally transparent, and the
polynomial bounds are machine-checked rather than merely claimed.

\textbf{Structure of the paper.} 
Section~\ref{sec:background} introduces CTF, CTS, and SVI. 
Section~\ref{sec:formalisation} describes our Rocq formalisation and proves 
the forward direction. Section~\ref{sec:tla} reports empirical validation 
results that show the reverse direction fails; the boundary is detailed in 
Section~\ref{sec:boundary}. Section~\ref{sec:vfr} presents VFR, validated 
experimentally in Section~\ref{sec:experiments}. Section~\ref{sec:discussion} discusses residual benefits, 
related work, limitations, and future directions. 
Section~\ref{sec:conclusion} concludes.

\subsection{Scope and limitations}
\label{sec:scope}

To avoid misunderstanding, we state the scope of the formalisation explicitly. Table~\ref{tab:trust-boundary} summarises every component of the pipeline, marking each as formally verified, trusted, or heuristic.
\textbf{Verified in Rocq:} CTF, CTS, clearing, aligned intersection, and the clause-by-clause CNF-to-CTF translation in which every clause becomes its own tier. For this fragment theorems about path existence, SVI soundness, and polynomial-time bounds are mechanically proved.

\textbf{Not verified:} the grouped-window decomposition that merges multiple clauses sharing the same variable triple into a single tier; the dense sliding-window, overlapping-group, and mixed-overlap benchmarks; and the general 3-CNF heuristic pipeline. These are empirical illustrations of an unverified heuristic, not formal results. See Section~\ref{sec:heuristic-pipeline} for the heuristic pipeline and Section~\ref{sec:experiments} for the benchmarks.

\begin{table}[h]
\centering
\small
\caption{Trust boundary: verified components, trusted base, and heuristics. ``Extr.'' indicates whether the component is extracted to executable OCaml code.}
\label{tab:trust-boundary}
\begin{tabularx}{\textwidth}{@{}Xl>{\raggedright\arraybackslash}p{3.8cm}c@{}}
\toprule
Component & Status & Reference / Caveat & Extr. \\
\midrule
CTF, CTS, clearing definitions & Verified & Rocq 9.1.1, no admitted proofs & No \\
SVI soundness (forward) & Verified & Theorem~\ref{thm:forward-main} & No \\
Aligned intersection (bi-impl.) & Verified & Theorem~\ref{thm:compat} & No \\
Systemic aligned ($k$ structs) & Verified & Theorem~\ref{thm:compat-k} & No \\
Polynomial complexity & Verified & $O(n^4)$ clearing (generic), $O(n^2)$ clearing (single-forbidden) and SVI; $n$ = CTS size (tiers+triplets). Full solver includes exponential post-check & No \\
CNF$\to$CTF (clause-by-clause) & Verified & Sliding-window CNF only & Yes \\
Strong CTF predicate & Verified & GapClosure.v & No \\
Swansea RUP (Reverse Unit Propagation) checker & Trusted base & Rocq-extracted elsewhere & Yes \\
OCaml extraction & Trusted base & Rocq $\to$ OCaml compiler & N/A \\
Z3 SAT solver & Trusted external & UNSAT proofs checked by Swansea & N/A \\
Grouped-window decomposition & Heuristic & Forward soundness only; completeness open & No \\
Dense / overlap / mixed pipelines & Heuristic & Empirical evaluation only & No \\
General 3-CNF (Z3 fallback) & Heuristic & Unverified decomposition & No \\
\bottomrule
\end{tabularx}
\end{table}

\section{Background: Romanov's triplet logic}
\label{sec:background}

\begin{quote}\small
\textbf{Verified core (mechanised in Rocq).} For well-formed sliding-window CNF, the clause-by-clause translation, clearing, aligned intersection, and SVI filter are formally proved correct (Table~\ref{tab:trust-boundary}).\par
\vspace{2pt}
\textbf{Heuristic shell (unverified).} Grouped-window decomposition, greedy permutation search, overlapping-group handling, post-check backtracking, and the general 3-CNF pipeline are empirical heuristics with no formal guarantee.
\end{quote}

\subsection{Compact Triplets Formula (CTF)}

A 3-CNF formula over $n$ Boolean variables $x_1, \dots, x_n$ is a conjunction of 
clauses, each a disjunction of exactly three literals. In TLS, a clause is 
represented as a \emph{triplet} $(v_1, v_2, v_3) \in \{0,1\}^3$ relative to an 
ordered triple of variable indices $(i, j, k)$. A \emph{Compact Triplets Formula} 
(CTF) is a collection of such triplets grouped by their variable indices into 
\emph{tiers}.

\begin{definition}[Tier]
 A \emph{tier} over variable indices $(i, j, k)$ is a set of triplets 
 $t \subseteq \{0,1\}^3$. A \emph{CTF} is a list of tiers.
\end{definition}

The pipeline translates each 3-literal clause into a forbidden triplet (negation pattern), builds a CTF tier per clause, and then forms the raw CTS by tier-wise complementation. The clearing procedure (Listing~\ref{lst:clearing}) is applied last to remove incompatible triplets.

\subsection{Compact Triplets Structure (CTS)}

Given a CTF $F$, the \emph{Compact Triplets Structure} $S = \mathrm{CTS}(F)$ is 
obtained by replacing each tier $t$ of $F$ with its complement:
\begin{equation}\label{eq:complement}
 S_i = \{0,1\}^3 \setminus F_i.
\end{equation}
Intuitively, $S$ contains all triplets that are \emph{not} forbidden by the 
corresponding clause group.

\begin{definition}[Compatibility]
 Two triplets $a = (a_1, a_2, a_3)$ and $b = (b_1, b_2, b_3)$ are 
 \emph{compatible} if their overlapping positions agree:
 \begin{equation}\label{eq:compatible}
 a_2 = b_1 \quad \text{and} \quad a_3 = b_2.
 \end{equation}
\end{definition}

\begin{lemma}[Compatibility Degree]
 \label{lem:degree}
 For every triplet $t \in \{0,1\}^3$ there are exactly two triplets $t'$
 with $\mathrm{compatible}(t, t') = \mathsf{true}$ (forward) and exactly two with
 $\mathrm{compatible}(t', t) = \mathsf{true}$ (backward).
\end{lemma}

Consequently, the clearing procedure cannot remove a triplet because it has
``too many'' neighbours; rather, it removes triplets whose two potential
partners have already been eliminated. For a fixed triplet $t = (0,1,1)$,
exactly two triplets are compatible in the forward direction and exactly two in
the reverse direction (Figure~\ref{fig:compatible-degree}).

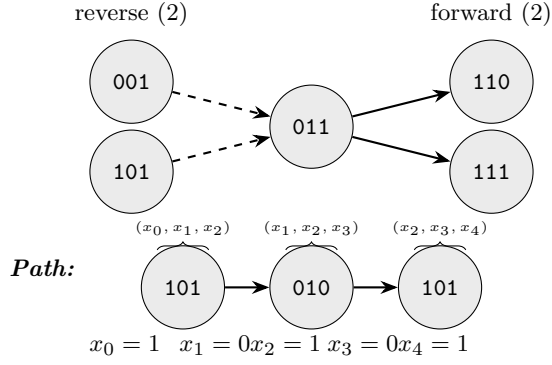
\begin{figure}[ht]
\centering
\begin{tikzpicture}[scale=0.85,
 tripletnode/.style={draw, circle, minimum size=1.1cm, align=center, fill=black!8, font=\small\ttfamily},
 varnode/.style={font=\scriptsize\itshape}]
 
 \node[tripletnode] (center) at (0,2.5) {011};
 \node[tripletnode] (f1) at (2.8,3.2) {110};
 \node[tripletnode] (f2) at (2.8,1.8) {111};
 \node[tripletnode] (b1) at (-2.8,3.2) {001};
 \node[tripletnode] (b2) at (-2.8,1.8) {101};
 
 \draw[-{Stealth}, thick] (center) -- (f1);
 \draw[-{Stealth}, thick] (center) -- (f2);
 \draw[-{Stealth}, thick, dashed] (b1) -- (center);
 \draw[-{Stealth}, thick, dashed] (b2) -- (center);
 
 \node[above=0.05cm of f1, font=\scriptsize] {forward (2)};
 \node[above=0.05cm of b1, font=\scriptsize] {reverse (2)};
 
 \node[varnode] at (-4.2,0.3) {\textbf{Path:}};
 
 \node[tripletnode] (p0) at (-2,0) {101};
 \node[tripletnode] (p1) at (0,0) {010};
 \node[tripletnode] (p2) at (2,0) {101};
 
 \draw[-{Stealth}, thick] (p0) -- (p1);
 \draw[-{Stealth}, thick] (p1) -- (p2);
 
 \node[varnode] at (-2.9,-0.9) {$x_0=1$};
 \node[varnode] at (-1.5,-0.9) {$x_1=0$};
 \node[varnode] at (-0.4,-0.9) {$x_2=1$};
 \node[varnode] at (0.8,-0.9) {$x_3=0$};
 \node[varnode] at (1.9,-0.9) {$x_4=1$};
 
 \draw[decorate,decoration={brace,amplitude=3pt,raise=2pt}] (-2.4,0.5) -- (-1.6,0.5) node[midway,above=4pt,font=\tiny] {$(x_0,x_1,x_2)$};
 \draw[decorate,decoration={brace,amplitude=3pt,raise=2pt}] (-0.4,0.5) -- (0.4,0.5) node[midway,above=4pt,font=\tiny] {$(x_1,x_2,x_3)$};
 \draw[decorate,decoration={brace,amplitude=3pt,raise=2pt}] (1.6,0.5) -- (2.4,0.5) node[midway,above=4pt,font=\tiny] {$(x_2,x_3,x_4)$};
\end{tikzpicture}
\caption{Top: the compatible-degree property for triplet $t=(0,1,1)$. Exactly two 
triplets are compatible in each direction (2-regular relation). Bottom: a valid 
path $101 \to 010 \to 101$ through three tiers, inducing the assignment 
$x_0=1, x_1=0, x_2=1, x_3=0, x_4=1$. Overlapping positions enforce consistency 
across adjacent triplets. The path is drawn in assignment order (increasing 
tier indices); \texttt{build\_paths\_all} stores paths in reverse order, 
prepending each new triplet ahead of the current head 
(Listing~\ref{lst:build-paths}).}
\label{fig:compatible-degree}
\end{figure}

A \emph{path} through a CTS $S = [t_1, \dots, t_m]$ is a sequence of triplets 
$p = [c_1, \dots, c_m]$ such that $c_i \in t_i$ and adjacent triplets are 
compatible. Each full-length path induces a variable assignment by flattening the 
triplets: if $p = [(a_1,b_1,c_1), (a_2,b_2,c_2), \dots]$, the corresponding 
\emph{satisfying set} is the list 
$\mathit{ss} = [a_1; b_1; c_1; a_2; b_2; c_2; \dots]$.
A list $\mathit{ss}$ \emph{satisfies} a CTS $S$ if every consecutive triple 
$(\mathit{ss}[3i], \mathit{ss}[3i+1], \mathit{ss}[3i+2])$ belongs to tier $i$.
A \emph{joint satisfying set} (JSS) of two structures $S_1, S_2$ is a single 
list $\mathit{ss}$ that satisfies both simultaneously.

Figure~\ref{fig:levels} summarises the abstraction stack. The construction 
pipeline transforms a 3-CNF formula (bottom) into a concrete assignment (top) 
through a sequence of representation changes; the dashed arrow shows the 
verified feedback loop.

\begin{figure}[ht]
\centering
\begin{tikzpicture}[
 scale=0.85,
 levelbox/.style={draw, rectangle, rounded corners, minimum width=0.85\linewidth, text width=0.85\linewidth, minimum height=0.7cm, align=center},
 arrowlabel/.style={font=\scriptsize\itshape}
]

\node[levelbox, fill=black!5] (cnf) at (0,0) {\textbf{Input:} 3-CNF formula (clauses over $x_0 \dots x_{n-1}$)};
\node[levelbox, fill=black!12] (ctf) at (0,1.4) {\textbf{Level 0:} CTF --- forbidden triplets (clause $\mapsto$ negation)};
\node[levelbox, fill=black!8] (cts) at (0,2.8) {\textbf{Level 1:} Raw CTS --- allowed triplets (complement of CTF)};
\node[levelbox, fill=black!16] (clear) at (0,4.2) {\textbf{Level 2:} Cleared CTS --- two-way adjoinable (SVI / EP filter)};
\node[levelbox, fill=black!20] (path) at (0,5.6) {\textbf{Level 3:} Path --- compatible sequence ($a_2=b_1,\; a_3=b_2$)};
\node[levelbox, fill=black!24] (assign) at (0,7.0) {\textbf{Output:} Assignment --- variable values (flattened path)};

\draw[-{Stealth}, thick] (cnf) -- node[right, arrowlabel] {cnf\_to\_ctf} (ctf);
\draw[-{Stealth}, thick] (ctf) -- node[right, arrowlabel] {complement} (cts);
\draw[-{Stealth}, thick] (cts) -- node[right, arrowlabel] {clear\_structure} (clear);
\draw[-{Stealth}, thick] (clear) -- node[right, arrowlabel] {build\_paths\_all} (path);
\draw[-{Stealth}, thick] (path) -- node[right, arrowlabel] {ss\_to\_assignment} (assign);

\draw[-{Stealth}, thick, dashed] (assign.east) -- ++(0.25cm,0) |- node[pos=0.25, right, arrowlabel] {eval\_cnf} (cnf.east);

\end{tikzpicture}
\caption{Levels of abstraction in VFR. Each layer transforms the 
representation toward a concrete variable assignment; the dashed arrow 
shows the verified feedback loop ~(OCaml-extracted \texttt{eval\_cnf}).}
\label{fig:levels}
\end{figure}
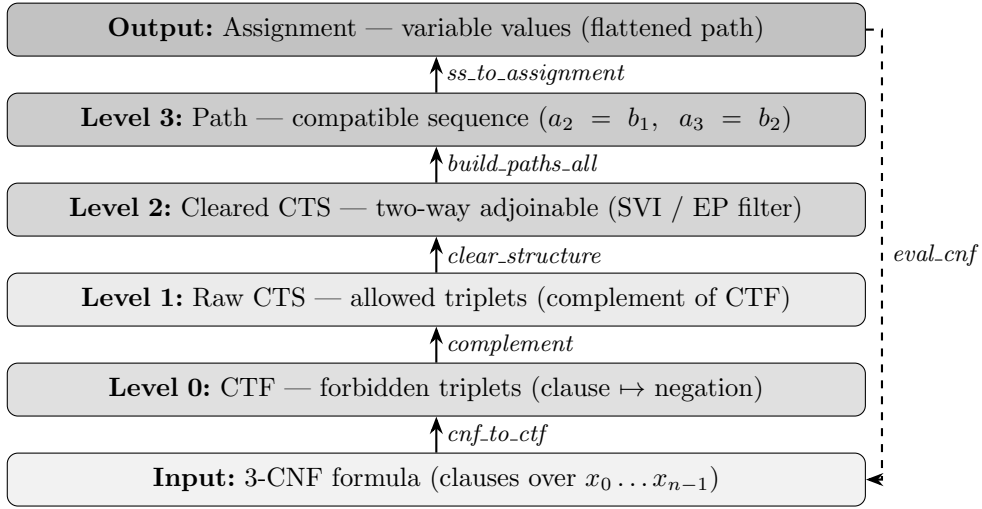

Each tier contains triplets over a sliding window of three variables; 
adjacent tiers overlap by two variables, ensuring that compatibility propagates 
constraints forward.

\subsection{Simple Vertex Intersection (SVI)}

For two CTS $S_1$ and $S_2$, the \emph{Simple Vertex Intersection} constructs 
a hyperstructure $H = \mathrm{SVI}(S_1, S_2)$ as follows:

\begin{enumerate}
 \item Build basic graphs $G_1$ and $G_2$ where vertices are triplets annotated 
 with their tier index.
 \item Compute the set of \emph{common vertices}: triplets that appear in 
 $G_1$ and whose triplet value appears somewhere in $G_2$ (not necessarily at 
 the same tier index).
 \item Return the hyperstructure containing these common vertices.
\end{enumerate}

In Romanov's framework, $H$ is non-empty if and only if $S_1$ and $S_2$ share a 
joint satisfying set---an assignment that satisfies both structures tier-by-tier
(a claim we refute in Section~\ref{sec:boundary}).

The definition extends naturally to $k \ge 2$ structures. The \emph{Systemic 
Simple Vertex Intersection} ($\mathrm{SSVI}$) computes common vertices across 
all pairs of structures in a family $\mathcal{S} = \{S_1,\dots,S_k\}$, producing 
a \emph{hyperstructure system} $\mathit{HSS}$. Non-emptiness of $\mathit{HSS}$ 
implies that every pair of structures in $\mathcal{S}$ shares a common triplet, 
which is the analogue of Theorem~\ref{thm:forward-main} for multiple structures (Theorem~\ref{thm:systemic-main} below).

\paragraph{From CNF to CTS.} The concrete construction algorithms---grouping clauses by variable sets, building complement tiers, and clearing---are described in Section~\ref{sec:heuristic-pipeline}. Only the clause-by-clause pipeline is formally verified (Table~\ref{tab:trust-boundary}).

\section{Heuristic Pipeline}
\label{sec:heuristic-pipeline}

\textbf{Caveat.} The pipeline described below---grouping clauses by variable 
sets and constructing CTS tiers via complementation---is the one used in the 
VFR \emph{prototype}. It is \textbf{not} formally verified in Rocq; only the 
simplified clause-by-clause pipeline of Section~\ref{sec:perm} (for well-formed 
sliding-window CNF) is mechanised. The following description serves as 
operational documentation and motivation for the verified fragment.

We now describe the concrete algorithms for constructing CTS from a 3-CNF 
formula. The pipeline consists of three phases: \emph{decomposition}, 
\emph{tier construction}, and \emph{clearing}.

\paragraph{Step 1: Decomposition.}
Given a 3-CNF formula $\phi$ with $n$ variables and $m$ clauses, group the 
clauses by their sets of variable indices. For each group $g$ with variables 
$\{i, j, k\}$ and $|g|$ clauses:
\begin{enumerate}
 \item Create a permutation $\pi = [i, j, k, \dots]$ where the first three 
 positions are the group's variables and the remaining $n-3$ positions are the 
 other variables in some fixed order.
 \item For each clause $C = (\ell_i \lor \ell_j \lor \ell_k)$ in the group, 
 encode it as a triplet $(v_1, v_2, v_3)$ where $v_p = 1$ if the $p$-th literal 
 is negated and $v_p = 0$ otherwise.
 \item Collect all triplets into a CTF $F_g$ annotated with variable indices 
 $(i, j, k)$.
\end{enumerate}
The result is a list of CTFs $[F_1, \dots, F_k]$ where $k \leq m$.

\footnote{The formalisation in \texttt{FormulaTranslation.v} translates each clause to a distinct tier for well-formed sliding-window CNF. The grouped-window pipeline is an unverified heuristic for general 3-CNF (Table~\ref{tab:trust-boundary}).}

\paragraph{Step 2: Tier Construction.}
For each CTF $F$ with tiers grouped by variable indices $(i, j, k)$:
\begin{enumerate}
 \item For each tier $t$ containing forbidden triplets $T \subseteq \{0,1\}^3$, 
 construct the complement tier $t' = \{0,1\}^3 \setminus T$.
 \item Assemble the tiers into a raw CTS $S_{\text{raw}}$.
\end{enumerate}

\paragraph{Step 3: Clearing Procedure.}
The raw CTS may contain incompatible triplets that cannot participate in any 
full-length path. The \emph{clearing procedure} removes them iteratively
(Listing~\ref{lst:clearing}):

\begin{lstlisting}[frame=single,caption={Clearing Procedure (pseudocode)},label={lst:clearing}]
function ClearSingleTier(tier_idx, all_tiers):
    current  := all_tiers[tier_idx]
    prev     := all_tiers[tier_idx-1] if tier_idx > 0 else AllTriplets
    next     := all_tiers[tier_idx+1] if tier_idx+1 < |all_tiers| else AllTriplets
    return { t in current | exists p in prev : Compatible(t,p)
                         and exists n in next : Compatible(t,n) }

function ClearStructurePass(S):
    return [ ClearSingleTier(i, S) for i = 0 .. |S|-1 ]

function ClearStructure(S):
    repeat cts_size(S) times:
        S := ClearStructurePass(S)
    return S
\end{lstlisting}

Here \texttt{AllTriplets} denotes the set of all $2^3 = 8$ possible triplet values;
it is used as a boundary condition so that end tiers do not need special-casing.
The iteration bound $\mathrm{cts\_size}(S) = \sum_i |T_i|$ is the total number of
triplets; it suffices because each pass either removes at least one triplet or
leaves the structure unchanged (Theorem~\ref{thm:clear-term}).

This is a fixed-point computation: in each pass, a \emph{new} tier is built 
containing only triplets that have at least one compatible neighbour in each 
adjacent tier (or lie at a boundary). The process repeats until no more triplets 
are eliminated. The resulting structure is the \emph{cleared CTS}. Note that the 
implementation builds a fresh list (\texttt{[t for t in current if ...]}) using a functional list comprehension rather than in-place mutation.

\paragraph{Example.}
Consider the formula over 3 variables:
\[
 \phi = (x_1 \lor x_2 \lor x_3) \land 
 (\neg x_1 \lor x_2 \lor \neg x_3) \land
 (x_1 \lor \neg x_2 \lor x_3).
\]
All clauses share variables $\{1, 2, 3\}$, so decomposition yields a single CTF 
with one tier and forbidden triplets:
\[
 T = \{(0,0,0),\; (1,0,1),\; (0,1,0)\}.
\]
The complement tier contains the remaining 5 triplets:
\[
 t' = \{(0,0,1),\; (0,1,1),\; (1,0,0),\; (1,1,0),\; (1,1,1)\}.
\]
Since there is only one tier, the clearing procedure does nothing. The final 
CTS consists of this single tier.

\begin{lstlisting}[frame=single,caption={Path-Building Algorithm (pseudocode)},label={lst:build-paths}]
function ExtendPaths(partial_paths, tier):
    if partial_paths is empty:
        return [ [t] for t in tier ]
    result := []
    for path in partial_paths:
        for t in tier:
            if Compatible(t, head(path)):
                result.append( [t] + path )
    return result

function BuildPathsAll(cts):
    paths := []
    for tier in cts:
        paths := ExtendPaths(paths, tier)
    return [ p for p in paths if len(p) == len(cts) ]
\end{lstlisting}

\begin{definition}[Sliding-Window CNF]
\label{def:sliding}
A 3-CNF formula $\phi$ with $n$ variables is a \emph{well-formed sliding-window 
CNF} if its clauses can be ordered so that the $i$-th clause contains exactly 
the variables $(x_i, x_{i+1}, x_{i+2})$ for $i = 0, \dots, n-3$. Multiple 
clauses may share the same window.
\end{definition}

\begin{remark}[Tractability of the verified fragment]
Well-formed sliding-window CNF has primal graph pathwidth~$\le 2$: each clause
covers three consecutive vertices of a path, so the primal graph is a subgraph
of the square of a path. SAT for graphs of bounded pathwidth is a classic
tractable case: dynamic programming on a path decomposition solves it in linear
time. Consequently, the verified polynomial bounds for clearing and SVI
\emph{do not expand the class of polynomially solvable 3-SAT instances}; they
merely reconstruct the same pathwidth-$\le 2$ fragment within Romanov's
framework, but with a formally certified algorithm. The value lies in the
mechanisation---certifying that Romanov's triplet constructions yield a correct
decision procedure for this fragment---not in a new complexity improvement.
\end{remark}

\subsection{Why General 3-CNF Falls Outside the Verified Fragment}
\label{sec:perm}

The decomposition described above preserves satisfiability only when the
input formula is already a well-formed sliding-window CNF (Definition~\ref{def:sliding}),
i.e., every clause uses three consecutive variables. For an arbitrary 3-CNF
formula containing clauses such as $(x_1 \lor x_5 \lor x_9)$, typically no variable
ordering can place those three variables in consecutive positions while
simultaneously satisfying the same requirement for all other clauses.

Formally, the question ``does there exist a permutation of variables that
makes this 3-CNF sliding-window?'' is the \emph{consecutive-ones property} of
the clause--variable incidence matrix: columns correspond to variables, and
each clause-row must have its three ones in consecutive positions. This
property is decidable in polynomial time by the classical PQ-tree algorithm
of Booth and Lueker~\cite{booth1976consecutive}. The primal graph of a
sliding-window formula is an interval graph of clique number at most three
(and hence of pathwidth at most two), and bounded pathwidth is decidable in
linear time for every fixed width.

The limitation is therefore structural rather than algorithmic: the
sliding-window class is narrow, and a generic 3-CNF formula admits no valid
ordering. Moreover, a polynomial-time transformation of an arbitrary 3-CNF
into an equisatisfiable sliding-window CNF would imply P${}={}$NP, since the
fragment itself is solvable in polynomial time by dynamic programming on its
bounded-pathwidth primal graph. This justifies our restriction: our Rocq
theorems apply to sliding-window CNF formulas; general decomposition falls
outside the verified fragment (Table~\ref{tab:trust-boundary}).

\paragraph{Permutation search.}
Although the existence of a valid ordering is decidable in polynomial time
(see above), our \emph{verified} implementation is an exhaustive permutation
search (file \texttt{Permutation.v}) that guarantees
to find a variable ordering making the formula sliding-window whenever one
exists. The Rocq formalisation defines \texttt{permute\_cnf} (apply a variable
permutation) and a sliding-window predicate (checked via
\texttt{is\_sliding\_cnf\_bool}), together with a length-preservation lemma.
The verified algorithm
\texttt{exhaustive\_sliding\_window\_permutation} enumerates all $n!$
permutations and returns the first valid one; we prove both soundness (any
returned permutation is correct) and completeness (if a sliding-window ordering
exists, it is found). Because of factorial complexity the search is practical
only for $n \leq 8$ ($8! = 40\,320$). For larger instances the Python solver
falls back to an unverified heuristic backtracking search (module
\texttt{permutation\_heuristic.py}) or delegates directly to the external CDCL
pipeline. Permutations are proved to preserve satisfiability
(\texttt{permute\_cnf\_preserves\_sat}); a verified polynomial-time
recognition algorithm (replacing the exhaustive search) is left to future
work. When a permutation is found, the solver follows the
verified clause-by-clause pipeline on the permuted formula; when it fails, the
fallback path verified by external proof checking (Z3 + Swansea RUP checker) takes over. This yields a
``heuristic fallback'' path: for small structured instances the entire pipeline is
formally guaranteed, while for large random instances the trust boundary reduces
to the external proof checker.

\section{Formalisation in Rocq}
\label{sec:formalisation}

We formalised the core data structures and algorithms of TLS in Rocq~9.1.1. 
The development spans seventeen Rocq files covering triplets, tiers, compatibility, 
CTF-to-CTS construction, clearing, path building, SVI and systemic alignment, 
formula translation and equivalence, formal counterexamples, grouped sliding-window CNF, 
variable relabelling, overlapping groups, gap closure, heuristic decomposition, 
exhaustive permutation search, and OCaml extraction.

\subsection{Key Definitions}

A triplet is a triple of Booleans. A tier is a list of triplets. A CTS 
is a list of tiers. A satisfying set $\mathit{ss}$ is a flattened 
path: if the underlying path of triplets is $[(a_1,b_1,c_1), (a_2,b_2,c_2), \dots]$, 
then $\mathit{ss}$ $=$ $[a_1; b_1; c_1; a_2; b_2; c_2; \dots]$. The function 
$\mathrm{extract\_triplet\_local}(\mathit{ss},i)$ retrieves the $i$-th triplet as 
$(\mathit{ss}[3i], \mathit{ss}[3i+1], \mathit{ss}[3i+2])$ (or $\mathsf{None}$ if the list is too short). 
Consequently, for a satisfying set produced by $\mathrm{ss\_from\_path\_flat}$, 
$\mathrm{extract\_triplet\_local}(\mathit{ss},i)$ yields the sliding-window triplet over 
variables $(x_i, x_{i+1}, x_{i+2})$ (the flattening enforces $b_i = a_{i+1}$ and 
$c_i = b_{i+1}$, so the logical window slides by one variable even though the 
list indices advance by three).
The constructor $\mathrm{ss\_from\_path\_flat}$ flattens a path of triplets 
$[(a_1,b_1,c_1), (a_2,b_2,c_2), \dots]$ into 
$[a_1; b_1; c_1; a_2; b_2; c_2; \dots]$, so consecutive triplets share the 
overlapping variables $b_1=a_2$ and $c_1=b_2$ required by $\mathrm{compatible}$.

\begin{definition}[Satisfying Set]
 \label{def:satisfying-set}
 Let $S = [T_0, T_1, \dots, T_{m-1}]$ be a CTS and let $\mathit{ss} = [v_0, v_1, \dots]$ be a list of Boolean values. For each tier index $i$, let $\mathrm{trip}_i(\mathit{ss}) = (\mathit{ss}[3i], \mathit{ss}[3i+1], \mathit{ss}[3i+2])$ (or undefined if the list is too short). Then $\mathit{ss}$ \emph{satisfies} $S$, written $\mathrm{is\_satisfying\_set}(S, \mathit{ss})$, iff for every $i < m$ the triplet $\mathrm{trip}_i(\mathit{ss})$ is defined and belongs to $T_i$.
\end{definition}

\begin{definition}[Joint Satisfying Set]
 \label{def:jss}
 A list $\mathit{ss}$ is a \emph{joint satisfying set} of $S_1$ and $S_2$, written $\mathrm{JSS}(S_1, S_2, \mathit{ss})$, iff it satisfies both structures simultaneously: $\mathrm{is\_satisfying\_set}(S_1, \mathit{ss}) \land \mathrm{is\_satisfying\_set}(S_2, \mathit{ss})$.
\end{definition}

An empty structure is vacuously satisfiable; non-emptiness of both structures is enforced as an explicit premise in Theorem~\ref{thm:forward-main} ($S_1 \neq [], S_2 \neq []$).

\subsection{Design Choices}

We represent triplets as native triples \texttt{(bool * bool * bool)} and 
structures as lists rather than vectors or finite types. This choice reflects 
a trade-off between expressiveness and proof automation: lists provide 
structural induction principles that Rocq's \texttt{auto} and \texttt{lia} tactics 
handle well, while avoiding the proof-engineering overhead of dependent types for 
fixed-length sequences.

A subtle issue arose with the definition of \texttt{path}. Our initial attempt 
used \texttt{Definition path := list triplet}, which created a type synonym 
that is convertible but not unifiable with \texttt{list triplet}. This blocked 
\texttt{rewrite} and \texttt{subst} tactics across the development. Removing the 
type synonym and using \texttt{list triplet} directly resolved the issue and is 
a lesson for other mechanisation efforts.

\subsection{Proved Lemmas}

We proved 424 lemmas and theorems across the seventeen files. The development required 
nested inductions and careful handling of arithmetic side conditions involving 
\texttt{nth\_error} and tier lengths. The most technically demanding was 
\texttt{build\_paths\_aux\_contains\_expected\_path}, which required nested 
induction on the accumulator of \texttt{build\_paths\_aux} together with delicate 
arithmetic reasoning about \texttt{nth\_error} and tier lengths. 
Key lemmas include:

\begin{itemize}
 \item \textbf{Path existence} 
 (Lemma~\ref{lem:path-from-ss}): if a satisfying set 
 exists, \texttt{build\_paths\_aux} contains a path extending any suffix with 
 triplets drawn from the set.
 
 \item \textbf{Forward direction (Theorem~\ref{thm:forward-main}):} existence of a joint satisfying set implies 
 non-emptiness of SVI.
 
 \item \textbf{Aligned intersection equivalence} 
 (Theorem~\ref{thm:compat}): for structures of equal length, the aligned 
 intersection has a full-length path iff a compatible joint satisfying set 
 exists.
 
 \item \textbf{Systemic aligned completeness} 
 (Theorem~\ref{thm:compat-k}): for a system of $k$ aligned structures, the 
 systemic tier-wise intersection has a full-length path iff a compatible joint 
 satisfying set exists for the entire system.

 \item \textbf{CTF-to-CTS soundness} 
 (Lemma~\ref{lem:ctf-sound}): every path produced by \texttt{build\_paths\_all} 
 on \texttt{ctf\_to\_cts} yields a satisfying assignment for both the original 
 formula and the cleared structure.

 \item \textbf{Clearing fixed-point characterisation} 
 (Theorem~\ref{thm:clear-fixed}): after clearing, 
 every remaining triplet has compatible neighbours in both adjacent tiers 
 (or lies at a boundary).

 \item \textbf{Raw construction completeness} 
 (Lemma~\ref{lem:raw-complete}): if a satisfying set is 
 compatible with the raw (non-cleared) CTS, then \texttt{build\_paths\_all} 
 on that raw CTS is non-empty.

 \item \textbf{Cleared CTF strong completeness} 
 (Lemma~\ref{lem:cleared-strong}): for a non-empty cleared CTF, 
 compatibility-aware satisfiability implies non-emptiness of 
 \texttt{build\_paths\_all}. This closes the loop between the strong 
 predicate and path existence for the cleared structure.

 \item \textbf{Semantic gap between weak satisfiability and path 
 existence} (Theorem~\ref{thm:gap}): there exist CTFs 
 that admit locally consistent assignments under $\mathrm{satisfies\_ctf}$ 
 yet yield no compatible path after clearing. This shows that the weak 
 predicate does not guarantee global path existence; completeness is 
 recovered only at the level of aligned intersection 
 (Theorem~\ref{thm:compat}).

 \item \textbf{CNF-to-CTF satisfiability equivalence} 
 (Theorem~\ref{thm:cnf-ctf-equiv}): for well-formed 
 sliding-window CNF formulas, satisfiability in standard CNF semantics is 
 equivalent to satisfiability of the translated CTF. The formalisation 
 uses a \emph{simplified clause-by-clause pipeline}: each clause becomes 
 a separate CTF tier containing exactly one forbidden triplet (its 
 negation pattern). The proof constructs a greedy assignment 
 \texttt{sat\_assignment\_aux} that satisfies each clause independently. 
 Because the translation is one-to-one, consecutive tiers necessarily 
 share overlapping variables, and the satisfying set 
 $\mathrm{ss\_from\_path\_flat}$ enforces global consistency. This is 
 \emph{not} the grouped-window pipeline of Section~\ref{sec:heuristic-pipeline}; it 
 is an equivalent simplified view used for formal verification.

 \item \textbf{Path existence equivalence} 
 (Lemma~\ref{lem:struct-char}): a compatible 
 satisfying set exists for a non-empty structure iff 
 \texttt{build\_paths\_all} contains at least one path.

 \item \textbf{Raw construction emptiness} 
 (Lemma~\ref{lem:empty-tier}): a tier of 
 \texttt{build\_cts\_\allowbreak from\_ctf} is empty iff the corresponding formula 
 tier contains all 8 triplets\allowbreak (per-tier complementation).

 \item \textbf{Compatible degree} 
 (Lemma~\ref{lem:degree}): 
 for any triplet, exactly 2 triplets are compatible in the forward 
 direction and exactly 2 in the reverse direction.
\end{itemize}

\begin{theorem}[Forward Direction of SVI]
 \label{thm:forward-main}
 Let $S_1, S_2$ be non-empty CTS. If there exists a joint satisfying set $\mathit{ss}$ for $S_1$ and $S_2$ (in the weak, index-flexible sense), then the Simple Vertex Intersection is non-empty:
 \[
 \exists \mathit{ss} \; . \; \mathrm{JSS}(S_1, S_2, \mathit{ss})
 \;\Longrightarrow\;
 \mathrm{SVI}(S_1, S_2) \neq \emptyset.
 \]
\end{theorem}

\begin{proof}[Proof sketch]
 Every triplet $t_i$ of $\mathit{ss}$ appears in some tier of both $S_1$ and $S_2$. SVI computes common vertices by value, so at least $t_0$ is a common vertex. Hence the hyperstructure contains $t_0$ and is non-empty.
\end{proof}

\begin{theorem}[Systemic Forward Direction]
 \label{thm:systemic-main}
 Let $\mathcal{S} = [S_1, \dots, S_k]$ be a system of non-empty CTS. If there exists a joint satisfying set for all structures in $\mathcal{S}$, then the systemic SVI produces a non-empty hyperstructure system:
 \[
 \exists \mathit{ss} \; . \; \forall i \leq k \; . \;
 \mathrm{is\_satisfying\_set}(S_i, \mathit{ss})
 \;\Longrightarrow\;
 \mathrm{SSVI}(\mathcal{S}) \neq \emptyset.
 \]
\end{theorem}

\begin{proof}[Proof sketch]
 SSVI applies the pairwise SVI to every pair $(S_1, S_i)$ with $i > 1$. By Theorem~\ref{thm:forward-main} each pairwise SVI is non-empty because $\mathit{ss}$ satisfies both structures. Therefore every hyperstructure in the system is non-empty.
\end{proof}

\begin{theorem}[Aligned Intersection Equivalence]
\label{thm:compat}
For all $S_1, S_2$ with $|S_1| = |S_2|$, if both are non-empty, then
\[
 \exists \text{ JSS}_{\mathrm{compat}}(S_1, S_2)
 \;\Longleftrightarrow\;
 \mathrm{build\_paths\_all}(\mathrm{cts\_intersection\_raw}(S_1, S_2)) \neq [].
\]
\end{theorem}

Although the statement is intuitively clear---a compatible path through the
tier-wise intersection exists precisely when a sequence of triplets is
compatible across \emph{both} structures---it is exactly the kind of
``obvious'' fact that often breaks during mechanisation. The value of the
formal proof is that it certifies consistency among three independently defined
concepts (the compatibility-aware satisfying set, the aligned intersection
operation, and the inductive path-building algorithm) and establishes that the
right-hand side is \emph{computable} and has been extracted to OCaml
(Section~\ref{sec:vfr}). Reconciling the recursive definitions of
$\mathrm{is\_satisfying\_set\_compat}$ and \texttt{build\_paths\_all} is a
non-trivial engineering task: subtle mismatches in base cases or accumulator
shapes that a human reader overlooks become hard proof obligations in Rocq.

\begin{lemma}[Path Extension]
 \label{lem:path-extension}
 Let $R$ be a suffix of a CTS, $\mathit{ss}$ a compatible satisfying set, and $\mathit{acc}$ a set of partial paths. If every path in $\mathit{acc}$ can be extended by triplets drawn from $\mathit{ss}$, then $\mathrm{build\_paths\_aux}(R, \mathit{acc})$ contains at least one full-length extension whose prefix consists exactly of the triplets prescribed by $\mathit{ss}$.
\end{lemma}

\begin{corollary}[Aligned Intersection Equivalence]
 \label{thm:compat-restated}
 For all $S_1, S_2$ with $|S_1| = |S_2|$, let $I = \mathrm{cts\_intersection\_raw}(S_1, S_2)$. Then
 \[
 \exists \mathit{ss} \; . \; \mathrm{JSS}_{\mathrm{compat}}(S_1, S_2, \mathit{ss})
 \;\Longleftrightarrow\;
 \mathrm{build\_paths\_all}(I) \neq [].
 \]
\end{corollary}

\begin{theorem}[Systemic Aligned Completeness]
\label{thm:compat-k}
Let $\mathcal{S} = [S_1, S_2, \dots, S_k]$ be a system of aligned CTS 
$($all of equal length$)$. Define the systemic raw intersection 
$I_k = \mathrm{cts\_intersection\_raw\_k}(\mathcal{S})$. Then
\[
 \exists \mathit{ss} \; . \; \forall i \leq k \; . \;
 \mathrm{is\_satisfying\_set\_compat}(S_i, \mathit{ss})
 \;\Longleftrightarrow\;
 \mathrm{build\_paths\_all}(I_k) \neq [].
\]
\end{theorem}

\begin{proof}[Proof sketch]
The proof proceeds by induction on the system size, using the pairwise 
aligned-intersection lemmas as the base case. For the forward direction, if a 
common satisfying set $\mathit{ss}$ exists, it satisfies every pairwise 
intersection, hence the systemic intersection, and therefore 
$\mathrm{build\_paths\_all}(I_k) \neq []$. For the reverse direction, any 
path through $I_k$ yields a sequence $\mathit{ss} = \mathrm{ss\_from\_path}(\pi)$ 
that belongs to every tier of every structure; by the splitting lemma this 
$\mathit{ss}$ satisfies each $S_i$ compatibly.
\end{proof}

\begin{lemma}[Raw Intersection Non-Empty]
 \label{lem:raw-intersect-nonempty}
 For non-empty aligned structures of equal length,
 $\mathrm{cts\_intersection\_raw}$ is never empty.
\end{lemma}

\paragraph{Proof sketch (Aligned Intersection).} 
The non-emptiness premise is redundant: for non-empty aligned structures,
$\mathrm{cts\_intersection\_raw}$ is never empty
(Lemma~\ref{lem:raw-intersect-nonempty}).
($\Rightarrow$) Given a compatible satisfying set $\mathit{ss}$, we construct a path 
through the intersection by induction on the tier index. At each step, the 
compatibility of adjacent triplets in $\mathit{ss}$ guarantees that the corresponding 
triplet exists in the intersection tier and is compatible with the previous 
one. By Lemma~\ref{lem:path-from-ss} (proved by nested 
induction on the recursive structure), this path is present in 
$\mathrm{build\_paths\_all}$.

($\Leftarrow$) Given a full-length path $p$ in the intersection, every 
triplet $c_i \in p$ belongs to tier $i$ of both $S_1$ and $S_2$. The 
compatibility of adjacent triplets in $p$ ensures that $\mathrm{ss\_from\_path}$ 
produces a satisfying set compatible with both structures. 

\begin{lemma}[Path Existence from Satisfying Set]
 \label{lem:path-from-ss}
 Let $S$ be a non-empty CTS and $\mathit{ss}$ a satisfying set for $S$.
 Then $\mathrm{build\_paths\_aux}$ contains a path extending any suffix with
 triplets drawn from $\mathit{ss}$.
\end{lemma}

\begin{lemma}[Structural Characterisation]
 \label{lem:struct-char}
 For any non-empty CTS $S$:
 \[
 \mathrm{build\_paths\_all}(S) = []
 \;\Longleftrightarrow\;
 \neg\,\exists \mathit{ss} \; . \; \mathrm{is\_satisfying\_set\_compat}(S, \mathit{ss}).
 \]
 Equivalently, $\mathrm{build\_paths\_all}(S) \neq []$ iff there exists at least
 one compatible satisfying set (hence at least one full-length path).
\end{lemma}

This equivalence closes the loop between the semantic notion of satisfying set
and the algorithmic notion of path existence for a \emph{single} structure.
Theorem~\ref{thm:compat} lifts the same equivalence to aligned intersections of
pairs.

\subsection{Translation Correctness}

The verified pipeline relies on two correctness bridges between the classical
CNF world and the triplet world.

\begin{lemma}[Soundness of CTF-to-CTS Translation]
 \label{lem:ctf-sound}
 For every formula $\phi$ and every path $p \in
 \mathrm{build\_paths\_all}(\mathrm{ctf\_to\_cts}(\phi))$, the flattened
 assignment $\mathrm{ss\_from\_path}(\mathrm{rev}(p))$ both satisfies $\phi$
 (in the formula-level sense) and is a compatible satisfying set for the
 cleared structure.
\end{lemma}

\begin{theorem}[CNF-to-CTF Equivalence]
 \label{thm:cnf-ctf-equiv}
 For well-formed sliding-window CNF formulas, satisfiability in standard CNF
 semantics is equivalent to satisfiability of the translated CTF.
\end{theorem}

Lemma~\ref{lem:ctf-sound} bridges formula-level satisfiability and
structure-level compatibility: any path produced by
$\mathrm{build\_paths\_all}$ on the translated CTS yields a satisfying
assignment for the original formula. Theorem~\ref{thm:cnf-ctf-equiv} certifies
that the entire translation pipeline (CNF$\to$CTF$\to$CTS$\to$paths) is
semantically faithful for the verified fragment.

\begin{lemma}[Per-Tier Complementation]
 \label{lem:empty-tier}
 A tier of the raw CTS is empty iff the corresponding formula tier contains
 all 8 triplets (i.e., the forbidden set covers the entire space).
\end{lemma}

This per-tier complementation is the mechanism behind the counterexample of
Section~\ref{sec:boundary} (combined with arc-consistency filtering in clearing).

\begin{lemma}[Raw Construction Completeness]
 \label{lem:raw-complete}
 For any non-empty formula $\phi$, if $\mathit{ss}$ is a compatible satisfying
 set of the raw CTS $\mathrm{build\_cts\_from\_ctf}(\phi)$, then
 $\mathrm{build\_paths\_all}$ on that raw CTS is non-empty.
\end{lemma}

Thus completeness holds for the raw (non-cleared) construction: any compatible
satisfying set guarantees a non-empty path set. The non-emptiness precondition
excludes the degenerate case of an empty formula.

\begin{lemma}[Cleared CTF Strong Completeness]
 \label{lem:cleared-strong}
 For a non-empty cleared CTF, compatibility-aware satisfiability\allowbreak implies
 non-emptiness\allowbreak of\allowbreak $\mathrm{build\_paths\_all}$.
\end{lemma}

\begin{theorem}[Relabelling preserves satisfiability]
 \label{thm:relabel-sat-short}
 Let $f$ be a formula in which every clause uses only variables
 $\{v_1,v_2,v_3\}$. Then
 \[
 \begin{array}{@{}l}
 (\exists a,\; \mathrm{eval\_cnf}(a,f) = \mathsf{true})\\
 \quad \Longleftrightarrow
 (\exists a',\; \mathrm{eval\_cnf}(a',\mathrm{relabel\_cnf}(v_1,v_2,v_3,f)) = \mathsf{true}).
 \end{array}
 \]
\end{theorem}

The forward direction builds $a$ from $a'$ by placing the three bits
$a'_0,a'_1,a'_2$ at positions $v_1,v_2,v_3$. The backward direction projects to
the first three positions.

\begin{lemma}[Dense Groups Soundness]
 \label{lem:dense-groups}
 For dense sliding-window groups, the forward CTF-level link \emph{soundness}
 holds: every assignment that satisfies the grouped CNF yields a satisfying set
 for the corresponding CTF.
\end{lemma}

\subsection{Counterexamples and Formal Limits}\label{sec:counterexamples}

The positive results above are complemented by two formal counterexamples that 
mark the exact boundaries of what the framework can guarantee.

\paragraph{Semantic gap.}
The weak predicate $\mathrm{satisfies\_ctf}$ checks each 3-bit window
independently and does \emph{not} enforce agreement on overlapping bits between
consecutive windows. Consequently it admits locally consistent assignments that
are not globally realisable as compatible paths.

\begin{theorem}[Semantic Gap]
 \label{thm:gap}
 There exists a CTF $\phi$ and an assignment $\mathit{ss}$ such that
 $\mathit{ss}$ satisfies $\phi$ in the weak formula-level sense, yet
 $\mathrm{build\_paths\_all}(\mathrm{ctf\_to\_cts}(\phi)) = []$.
\end{theorem}

\begin{proof}[Counterexample]
 Take a 2-tier CTF where tier~0 forbids all triplets except $(0,1,1)$ and
 tier~1 forbids all except $(1,0,1)$. These two survivors are incompatible
 (the overlap bits $1 \neq 0$ do not match), so clearing removes both and
 yields empty tiers. Yet the assignment $\mathit{ss} = [0,1,1,1,0,1]$ avoids
 each tier's local forbidden set, giving
 $\mathrm{satisfies\_ctf}(\mathit{ss}, \phi) = \text{true}$. This formally
 separates weak formula-level satisfiability from structure-level path
 existence.
\end{proof}

\paragraph{Clearing is not conservative.}
Clearing uses the directed predicate $\mathrm{can\_adjoin}$, which requires a
forward-compatible successor, whereas $\mathrm{build\_paths\_all}$ checks only
backward compatibility. Hence a triplet may belong to a valid path yet still be
removed.

\begin{theorem}[Clearing removes valid paths]
 \label{thm:clearing-not-conservative}
 There exists a CTS $S$ and a sequence $\mathit{ss}$ such that
 \begin{enumerate}
 \item $\mathrm{is\_satisfying\_set\_compat}(S, \mathit{ss}) = \mathsf{true}$,
 \item $\mathrm{build\_paths\_all}(S) \neq \varnothing$,
 \item $\mathrm{build\_paths\_all}(\mathrm{clear\_structure}(S)) = \varnothing$.
 \end{enumerate}
\end{theorem}

\begin{proof}[Counterexample]
Take a 3-tier CTS
\[
 S = [\{(1,0,0)\},\; \{(0,1,0)\},\; \{(0,0,1)\}]
\]
with the unique compatible path $\pi = [(1,0,0), (0,1,0), (0,0,1)]$.
The middle triplet $(0,1,0)$ has no forward-compatible successor in tier~2
(the only candidate $(0,0,1)$ is incompatible because the overlapping bit
$1 \neq 0$). Hence clearing removes $(0,1,0)$, which in turn destroys the
only full-length path. Yet $\pi$ satisfies every tier of $S$ compatibly, so
$\mathrm{build\_paths\_all}(S) \neq \varnothing$ while
$\mathrm{build\_paths\_all}(\mathrm{clear\_structure}(S)) = \varnothing$.
\end{proof}

\begin{remark}[Why these limits do not invalidate the main results]
 Clearing is nevertheless \emph{sound} (it never creates new paths) and
 \emph{monotone} in the reverse direction: any path that survives clearing
 was already present in the original structure (Lemma~\ref{lem:clearing-monotone}).
 Completeness is recovered by Theorem~\ref{thm:compat}, which bypasses clearing
 entirely and works directly on the aligned intersection. The semantic gap is
 closed for the verified fragment by the strong predicate
 $\mathrm{satisfies\_ctf\_strong}$ (Section~\ref{sec:boundary}).
\end{remark}

\paragraph{Positive properties of clearing.}
Despite the negative results above, clearing satisfies two fundamental
properties that justify its use as an optimisation.

\begin{theorem}[Clearing Termination]
 \label{thm:clear-term}
 The clearing procedure terminates. The measure
 $\mathrm{cts\_size}(S) = \sum_i |T_i|$ (the total number of triplets in
 all tiers) strictly decreases whenever triplets are eliminated and is bounded
 below by~$0$, so only finitely many eliminations are possible.
\end{theorem}

\begin{theorem}[Fixed-Point Characterisation]
 \label{thm:clear-fixed}
 After clearing, every surviving triplet $t$ in tier $i$ has at least one
 compatible predecessor in tier $i-1$ and at least one compatible successor in
 tier $i+1$ (except at the boundaries $i=0$ and $i=m-1$, where only the existing
 neighbour is required).
\end{theorem}

The termination bound uses $\mathrm{cts\_size}(S)$ rather than the coarse
$8 \cdot |S|$, yielding a tighter proof. The fixed-point theorem gives a
semantic characterisation of the survivors: only triplets with two-way
adjoinability remain.

\begin{lemma}[Clearing monotonicity for paths]
 \label{lem:clearing-monotone}
 For any non-empty CTS $S$ and any path $\pi$,
 \[
 \pi \in \mathrm{build\_paths\_all}(\mathrm{clear}(S))
 \;\Longrightarrow\;
 \pi \in \mathrm{build\_paths\_all}(S).
 \]
\end{lemma}

\begin{proof}[Proof sketch]
Construct the satisfying set $\mathit{ss} = \mathrm{ss}(\mathrm{rev}(\pi))$ from
the reversed path. Because $\pi$ is valid in $\mathrm{clear}(S)$, $\mathit{ss}$
is a compatibility-aware satisfying set for $\mathrm{clear}(S)$. Since clearing
only shrinks tiers, every triplet of $\mathit{ss}$ in a cleared tier also lies in
the original tier. By induction on the tier list, compatibility-awareness
transfers from the cleared structure to the original one, and the
path-existence theorem yields a path $\pi'$ with exactly the same triplets as
$\pi$.
\end{proof}

\subsection{Constructive Solver for Grouped Sliding-Window CNF}

For grouped sliding-window CNFs with disjoint variable ranges (multiple clauses
per consecutive variable window) we formalise in Rocq a recursive, extractable solver
$\texttt{solve\_grouped\_sliding}$. The solver processes each group
independently: it translates the group to a CTF, finds satisfying sets via the
verified path-building pipeline, and merges the resulting assignments into a
global assignment. If the pipeline yields no valid set for a group, the solver
falls back to a verified greedy assignment constructor
$\texttt{sat\_assignment\_aux}$.

\begin{theorem}[Grouped solver soundness]
 \label{thm:solve-grouped-sound}
 For every well-formed grouped sliding-window CNF, if
 $\texttt{solve\_grouped\_sliding}$ returns an assignment, that assignment
 satisfies the entire concatenated formula.
\end{theorem}

\begin{theorem}[Grouped solver completeness]
 \label{thm:solve-grouped-complete}
 For every well-formed grouped sliding-window CNF, if every group is
 individually satisfiable, then $\texttt{solve\_grouped\_sliding}$ returns
 some assignment.
\end{theorem}

The key ingredient is Lemma~\texttt{eval\_cnf\_merge\_assignments\_concat}
($\texttt{Structured.v}$): merging a group assignment into the accumulated
result preserves satisfiability of the concatenated formula, even when the
current group appears again later in the list. This allows the recursive solver
to handle duplicate groups without backtracking.

\textbf{Extraction metrics.}
The solver is not merely proved correct but extracted to executable OCaml code.
Rocq's \texttt{Extraction} command produces \texttt{VFR.ml} (786 lines, $\approx$21~KiB),
which contains the verified implementations of \texttt{build\_paths\_all},
\texttt{eval\_cnf}, \texttt{ctf\_to\_cts}, and \texttt{solve\_grouped\_sliding}.
A thin JSON bridge (\texttt{vfr\_solver.ml}, 261 lines, $\approx$9~KiB) reads CNF
formulas from stdin and prints satisfying assignments or UNSAT verdicts.
Together they form a standalone verified executable that requires no Rocq
runtime.

\section{Empirical Validation via Exhaustive Enumeration}
\label{sec:tla}

Before undertaking Rocq proofs of the reverse direction, we used a
Python-based model checker to exhaustively enumerate all pairs of CTS
structures up to small bounds (787{,}244 exhaustive cases plus 500 random
instances). Four SVI variants were checked; in every case where a reverse
implication was claimed, the model checker found a counterexample.
All counterexamples share the same pattern: SVI finds a common triplet between
tier $i$ of $S_1$ and tier $j \neq i$ of $S_2$, creating a non-empty
hyperstructure despite the absence of a tier-aligned satisfying set. This
empirical observation motivated the aligned-intersection construction of
Theorem~\ref{thm:compat}.

\section{The Correctness Boundary: False Positives in SVI}
\label{sec:boundary}

\subsection{The Correctness Boundary}

Romanov's framework assumes that non-emptiness of SVI is equivalent to the 
existence of a joint satisfying set. Formally:
\begin{equation}\label{eq:romanov-claim}
 \exists \text{ JSS}(S_1, S_2)
 \;\Longleftrightarrow\;
 H = \mathrm{SVI}(S_1, S_2) \neq \emptyset.
\end{equation}

Only the forward direction ($\Rightarrow$), i.e.
\begin{equation}\label{eq:forward}
 \exists\text{JSS}\;\Longrightarrow\;H\neq\emptyset,
\end{equation}
is true and is proved in our 
Rocq development as Theorem~\ref{thm:forward-main}. The converse 
($H\neq\emptyset\Rightarrow\exists\text{JSS}$) does not hold in general.
\footnote{The definition $\mathrm{is\_satisfying\_set}$ treats an empty 
structure as vacuously satisfiable ($[] \mapsto \text{True}$). Non-emptiness of 
both structures is therefore stated as an explicit premise in 
Theorem~\ref{thm:forward-main} ($S_1 \neq [], S_2 \neq []$) and in 
Theorem~\ref{thm:systemic-main} ($\forall S \in \text{system}, \text{structure}(S) \neq []$). 
The compat version $\mathrm{is\_satisfying\_set\_compat}$ also uses 
$[] \mapsto \text{True}$. Its completeness theorem 
(Theorem~\ref{thm:compat}) carries the premise 
$\mathrm{cts\_intersection\_raw} \neq []$, but this premise is 
redundant for non-empty aligned structures: 
$\mathrm{cts\_intersection\_raw}$ is never empty when both inputs are 
non-empty and aligned (Lemma~\ref{lem:raw-intersect-nonempty}).}

A separate observation concerns the \emph{semantic gap} between the weak 
formula-level predicate $\mathrm{satisfies\_ctf}$ and structure-level path 
existence. The predicate $\mathrm{satisfies\_ctf}$ checks each 3-bit window 
independently: it merely verifies that $\mathrm{extract\_triplet\_local}(\mathit{ss},i)$ 
avoids the forbidden triplets of tier~$i$. It does \emph{not} enforce that 
consecutive windows agree on their overlapping 2 bits. Consequently, 
$\mathrm{satisfies\_ctf}$ admits \emph{locally consistent} assignments that 
are not globally realisable as compatible paths.

We formalise this gap in Theorem~\ref{thm:gap}: there exists 
a CTF $\phi$ and an assignment $\mathit{ss}$ with 
$\mathrm{satisfies\_ctf}(\mathit{ss},\phi)=\text{true}$, yet 
$\mathrm{build\_paths\_all}(\mathrm{ctf\_to\_cts}(\phi))=[]$. 
In this counterexample, tier~0 permits only 
$(\mathit{false},\mathit{true},\mathit{true})$ and tier~1 permits only 
$(\mathit{true},\mathit{false},\mathit{true})$. These two survivors are 
incompatible (their overlapping bits do not match), so clearing 
removes both. The CTF is therefore \emph{unsatisfiable} as a sliding-window 
formula: the overlapping variable $x_2$ receives conflicting values 
($\mathit{true}$ from the first window, $\mathit{false}$ from the second). 
Yet $\mathrm{satisfies\_ctf}$ returns $\mathsf{true}$ because it inspects each tier 
through non-overlapping chunks $\mathit{ss}[3i\,..\,3i{+}2]$.

This does \emph{not} contradict Lemma~\ref{lem:ctf-sound}: 
that lemma states only the soundness direction (every path produced by 
\texttt{build\_paths\_all} yields a formula-level satisfying set), which holds 
for arbitrary CTFs, whereas the counterexample concerns the converse. 
Moreover, for well-formed sliding-window CNF translated via 
the clause-by-clause pipeline of \texttt{FormulaTranslation.v}, 
each clause becomes a tier with exactly one forbidden triplet, and the 
constructed satisfying set $\mathrm{ss\_from\_path\_flat}$ is a flattened 
assignment where consecutive triplets necessarily agree on overlaps. Hence 
$\mathrm{satisfies\_ctf}$ coincides with true satisfiability for such formulas. 
The counterexample operates outside this well-formed class (it is a generic 
CTF with 7 forbidden triplets per tier), exposing the weakness of the generic 
predicate.

It is important to note that \textbf{clearing is not a conservative 
transformation} with respect to arbitrary paths or strong satisfying sets. 
The predicate $\mathrm{can\_adjoin}$ used during clearing requires \emph{two-way} 
adjoinability (both a forward-compatible predecessor and a forward-compatible 
successor), whereas \texttt{build\_paths\_all} checks only backward compatibility 
(\texttt{compatible next prev}). Consequently, a triplet may belong to a valid 
path yet still be removed by clearing because it lacks a forward-compatible 
partner in the next tier (Theorem~\ref{thm:clearing-not-conservative}). The 
``emptiness'' in the counterexample above is a genuine loss of paths, not merely 
a detection of inconsistency. This directional asymmetry is an intrinsic 
feature of the current formalisation. It is not a bug: clearing is used only as
an optimisation (it is sound---it never creates new paths), while completeness is
recovered by Theorem~\ref{thm:compat} via aligned intersection and
\texttt{build\_paths\_all}, which do not rely on clearing preserving all paths.
This counterexample rules out the naive hope that running clearing before SVI
would yield a complete decision procedure; instead, completeness requires the
aligned-intersection construction of Theorem~\ref{thm:compat}.

At the same time, clearing \emph{is} monotone in the reverse direction:
any full-length path that survives clearing was already present in the
original structure (Lemma~\ref{lem:clearing-monotone}). Monotonicity does
not contradict non-conservativity: it merely states that clearing never
\emph{creates} new full-length paths, only \emph{removes} existing ones.
The aligned-intersection approach (Theorem~\ref{thm:compat}) recovers
completeness by operating on the intersection directly, bypassing the need
for clearing to preserve paths.

\subsection{Why SVI Gives False Positives}

The Simple Vertex Intersection computes common vertices between $G_1$ and $G_2$ without 
requiring \emph{tier alignment}. A triplet from tier $i$ of $S_1$ is considered 
``common'' if its value appears in \emph{any} tier of $S_2$. False positives arise 
from two distinct phenomena:
\begin{enumerate}
 \item \textbf{Cross-tier matching.} A triplet from tier $i$ of $S_1$ matches a 
 triplet from tier $j \neq i$ of $S_2$. SVI reports them as ``common'' even though 
 no aligned satisfying set can use them simultaneously.
 \item \textbf{Partial tier match.} Some corresponding tiers share triplets, but 
 at least one tier has no common triplet. SVI finds the existing matches and 
 reports non-empty, yet no aligned satisfying set exists because the unmatched 
 tier cannot be satisfied.
\end{enumerate}

For example, tier 0 of $S_1$ and tier 1 of $S_2$ share triplet $(0,0,0)$, so SVI reports
non-empty. However, tier 1 of $S_1$ and tier 1 of $S_2$ have no common triplet,
so no aligned satisfying set exists.

\paragraph{Bidirectional compatibility.}
Attempts to define forward-compatible satisfying sets and prove that
clearing preserves bidirectional compatibility were abandoned after we proved
that $\mathrm{compatible}$ is a directed relation and that path-derived satisfying
sets cannot guarantee forward compatibility. This does not affect our main
results: Theorem~\ref{thm:compat} establishes equivalence through
\texttt{build\_paths\_all}, which checks only backward compatibility.

\subsection{Concrete Counterexamples}

\paragraph{Cross-tier mismatch.}
Consider:
\[
 S_1 = [\{(0,0,0)\}, \{(1,1,1)\}], \quad
 S_2 = [\{(1,1,1)\}, \{(0,0,0)\}].
\]
SVI finds two common triplets: $(0,0,0)$ appears in tier 0 of $S_1$ and tier 1 of 
$S_2$; $(1,1,1)$ appears in tier 1 of $S_1$ and tier 0 of $S_2$. Consequently SVI 
returns a non-empty hyperstructure. Yet no aligned joint satisfying set exists, 
because tier 0 of $S_1$ and tier 0 of $S_2$ share no triplet, and likewise for 
tier 1. The structures are mutually ``shifted.''

\paragraph{Partial tier match.}
Consider:
\[
 S_1 = [\{(0,0,0)\}, \{(1,1,1)\}], \quad
 S_2 = [\{(0,0,0)\}, \{(0,1,0)\}].
\]
SVI finds the common triplet $(0,0,0)$ (present in tier 0 of both structures) and 
returns a non-empty hyperstructure. However, tier 1 of $S_1$ and tier 1 of $S_2$ 
have no common triplet, so no aligned joint satisfying set exists. The formula is 
unsatisfiable, but SVI reports non-empty.

\section{VFR: A Prototype Solver with Post-Checking}
\label{sec:vfr}

Since SVI is not a complete decision procedure, we propose \emph{VFR}: a 
prototype solver that uses SVI as a fast, one-sided filter for structured 
formulas, followed by an exhaustive post-check.

(All formal guarantees are summarised in Table~\ref{tab:trust-boundary}.)

\subsection{Architecture}

Figure~\ref{fig:vfr} illustrates the VFR solving pipeline. The input 
3-CNF formula is first decomposed into CTFs; each CTF is converted to a CTS 
via complementation and clearing. SVI then acts as a filter whose running time 
is polynomial in the structure size (see complexity analysis below): if it 
returns \texttt{False}, the formula is guaranteed unsatisfiable. If SVI returns 
\texttt{True}, the exponential post-check searches for a globally consistent 
assignment. Any candidate is finally verified against the original CNF.
(Only the SVI filter is formally verified; see Table~\ref{tab:trust-boundary}.)

\begin{figure}[ht]
\centering
\begin{tikzpicture}[
 node distance=0.6cm,
 proc/.style={draw, rectangle, rounded corners, minimum width=1.9cm, minimum height=0.45cm, align=center, fill=black!8, font=\footnotesize},
 dec/.style={draw, diamond, aspect=2.2, minimum width=1.4cm, minimum height=0.5cm, align=center, fill=black!10, font=\footnotesize},
 term/.style={draw, ellipse, minimum width=1.3cm, minimum height=0.4cm, align=center, fill=black!15, font=\footnotesize},
 io/.style={draw, trapezium, trapezium left angle=70, trapezium right angle=110, minimum width=1.5cm, minimum height=0.4cm, align=center, fill=black!5, font=\footnotesize},
 arrow/.style={-{Stealth[length=1.5mm]}, thick}
]
 \node[io] (in) {3-CNF};
 \node[proc, below=of in] (dec) {Decompose};
 \node[proc, below=of dec] (build) {Build CTS};
 \node[proc, below=of build] (clear) {Clear};
 \node[dec, below=of clear] (svi) {SVI $=\emptyset$?};
 \node[term, below left=0.45cm and 0.25cm of svi] (unsat1) {UNSAT};
 \node[proc, below right=0.45cm and 0.25cm of svi] (post) {Post-check};
 \node[dec, below=of post] (found) {Found?};
 \node[term, below left=0.45cm and 0.25cm of found] (sat) {SAT};
 \node[term, below right=0.45cm and 0.25cm of found] (unsat2) {UNSAT};

 \draw[arrow] (in) -- (dec);
 \draw[arrow] (dec) -- (build);
 \draw[arrow] (build) -- (clear);
 \draw[arrow] (clear) -- (svi);
 \draw[arrow] (svi) -- node[left, font=\tiny, pos=0.35] {Yes} (unsat1);
 \draw[arrow] (svi) -- node[right, font=\tiny, pos=0.35] {No} (post);
 \draw[arrow] (post) -- (found);
 \draw[arrow] (found) -- node[left, font=\tiny, pos=0.35] {Yes} (sat);
 \draw[arrow] (found) -- node[right, font=\tiny, pos=0.35] {No} (unsat2);
\end{tikzpicture}
\caption{The VFR solving pipeline as a flowchart. Only the SVI filter is formally 
verified for the sliding-window fragment; decomposition and post-check are unverified 
heuristics for general 3-CNF.}
\label{fig:vfr}
\end{figure}
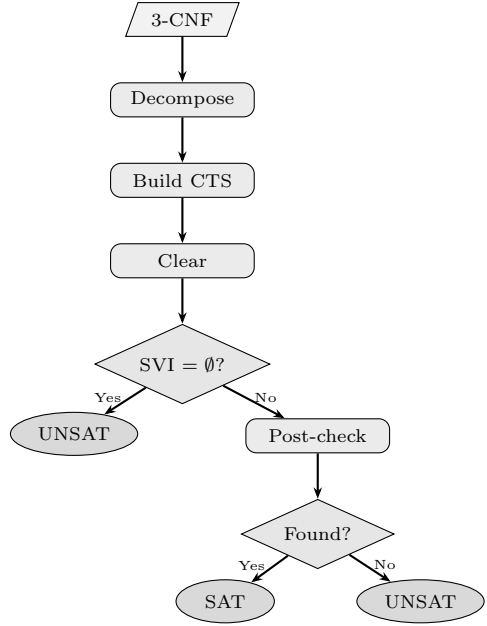

\begin{enumerate}
 \item \textbf{Decomposition.} The 3-CNF formula is split into CTFs grouped by 
 variable sets.
 
 \item \textbf{CTS Construction.} Each CTF is converted to a CTS via complement 
 and clearing.
 
 \item \textbf{SVI Filter (polynomial).} Run \texttt{EffectiveProcedure}. If it 
 returns \texttt{False}, return UNSAT immediately. This step is guaranteed 
 correct: no false negatives.
 
 \item \textbf{Post-Check (exponential).} If SVI returns \texttt{True}, search 
 for a globally consistent assignment by backtracking over partial assignments 
 from each CTS. If none exists, SVI produced a false positive; return UNSAT.
 
 \item \textbf{Verification.} Any candidate assignment is checked against the 
 original CNF formula.
\end{enumerate}

\paragraph{Complexity analysis.}
Table~\ref{tab:complexity} summarises the complexity of each pipeline stage.

\begin{table}[h]
\centering
\caption{Verified polynomial complexity bounds (mechanised in \texttt{Complexity.v} and \texttt{GapClosure.v})}
\label{tab:complexity}
\begin{tabularx}{\textwidth}{@{}Xll@{}}
\toprule
Stage & Time bound & Provenance \\
\midrule
Decomposition (heuristic) & $O(m \cdot \log m)$ & Grouping clauses by variable triples \\
Tier construction & $O(k)$ & $k$ tiers, at most 8 triplets each \\
Clearing (generic) & $\leq 100n^{4} + 100$ & Theorem~\ref{thm:clear-poly} \\
Clearing (single-forbidden) & $\leq 2000n^{2} + 2000$ & Theorem~\ref{thm:clear-poly-tight} \\
SVI filter & $\leq 3n^{2} + 7$ & Theorem~\ref{thm:svi-poly} \\
Full pipeline & $\leq 200n^{4} + 200$ & Theorem~\ref{thm:pipeline-poly} \\
\bottomrule
\end{tabularx}
\end{table}

The constants in Table~\ref{tab:complexity} emerge from a formal cost model
that counts every cons cell, compatibility check, and list traversal.
For generic clearing, each pass costs $O(n^3)$ (every triplet is checked against
both neighbours, each of size $O(n)$) and at most $n$ passes are iterated,
yielding $100n^4+100$ after routine arithmetic induction.
The SVI bound $3n^2+7$ comes from summing the quadratic costs of building two
basic graphs and intersecting their vertices.
For single-forbidden structures clearing is the identity
(Corollary~\ref{cor:clear-full-identity}), so only one pass is needed and the
bound collapses to $2000n^2+2000$.

\textbf{Decomposition} groups $m$ clauses by variable sets; sorting yields
$O(m \log m)$ time. \textbf{Tier construction} builds complement tiers of size
at most 8 triplets each, taking $O(k)$ for $k$ tiers. \textbf{Clearing} is a
fixed-point computation: each pass checks every line against all lines of each
neighbour, giving $O(t_j^2)$ per tier. The number of passes is bounded by the
total number of lines ($\mathrm{cts\_size}$), giving overall $O(\text{cts\_size} \cdot
\sum_j m_j t_j^2)$. In the mechanisation we use the exact measure
$\sum_i |\text{tier}_i|$ instead of the coarse $8 \cdot |S|$, which yields a
tighter termination proof (Theorem~\ref{thm:clear-term}).

\textbf{SVI filter} computes common vertices between all pairs of CTFs. For CTFs 
with $|V_j|$ and $|V_{j'}|$ vertices, the pairwise intersection is $O(|V_j| \cdot 
|V_{j'}|)$. Summing over all pairs gives $O(\sum_{j < j'} |V_j| \cdot |V_{j'}|)$, 
which is polynomial in the input size. For $k$ CTFs each with $O(n)$ triplets, 
this simplifies to $O(n^2 k^2)$.

\textbf{Formal verification.} The Rocq development includes a mechanised cost model (file \texttt{Complexity.v}, approximately 1,000 lines) that defines explicit cost functions mirroring the recursive structure of each algorithm. Every cons cell, compatibility check, and list traversal is assigned a unit cost. The bounds in Table~\ref{tab:complexity} are explicit inequalities proved in Rocq by induction on the structure of the input, where $n$ is the combined input size (tiers plus triplets). For the verified clause-by-clause translation each clause becomes one tier, so $n = \Theta(m)$ where $m$ is the number of clauses; for well-formed sliding-window CNF, $m = O(v)$ with $v$ the number of variables. These polynomial bounds apply \emph{only} to the filter stages (clearing and SVI). The complete decision procedure includes the aligned-intersection post-check, which is exponential in the worst case ($O(8^m)$ path enumeration). These are not merely 
asymptotic statements; they are explicit inequalities proved in Rocq by 
induction on the structure of the input, with arithmetic appeals to 
\texttt{nia}. Full definitions and intermediate lemmas (e.g., single-tier and single-pass bounds) are formalised in files \texttt{Complexity.v} and \texttt{GapClosure.v}.

\textbf{Aligned intersection} (used in the post-check) builds 
$\mathrm{cts\_intersection\_raw}$ tier-by-tier, then calls 
$\texttt{build\_paths\_all}$. The intersection construction is linear in tier 
size; path enumeration is the dominant cost: a CTS with $m$ tiers and at most 8 
triplets per tier has at most $8^m$ full-length paths in the worst case (all 
triplets mutually compatible). Hence $\texttt{build\_paths\_all}$ is 
$O(8^m)$ time and space.

\textbf{Post-check} backtracks over $k$ CTFs. If the $j$-th CTF has $p_j$ paths, the 
search explores the product space in $O(\prod_{j=1}^k p_j)$ time. Since $p_j 
\leq 8^{m_j}$, the worst-case matches aligned intersection. The recursion depth 
is $k$ and each partial assignment stores $n$ variables, giving $O(n \cdot k)$ 
space. This exponential worst case is unavoidable unless $\mathsf{P} = \mathsf{NP}$, but 
structured instances with tight constraints (small $t_{ij}$ after clearing) can 
be solved efficiently in practice.

\subsection{Derivation of Polynomial Bounds}
\label{sec:complexity-derivation}

The concrete constants in Table~\ref{tab:complexity} arise from a formal cost
model that assigns unit cost to every cons cell, compatibility check, and list
traversal. Generic clearing is quartic because each pass filters every triplet
against two neighbouring tiers and the process iterates at most
$\mathrm{cts\_size}(S)$ times; for single-forbidden sliding-window structures
clearing is the identity (Corollary~\ref{cor:clear-full-identity}), so only one
pass is needed and the bound collapses to quadratic. The SVI bound comes from
building two basic graphs and intersecting their vertices. Full cost-model
definitions, intermediate lemmas, and the arithmetic inductions are given in
files \texttt{Complexity.v} and \texttt{GapClosure.v}.

\begin{theorem}[Generic Clearing Bound]
 \label{thm:clear-poly}
 $\mathrm{clear\_structure}$ is bounded by $100 n^{4} + 100$ in the generic case.
\end{theorem}

\begin{theorem}[Tight Clearing Bound]
 \label{thm:clear-poly-tight}
 For single-forbidden sliding-window structures, $\mathrm{clear\_structure}$ is
 bounded by $2000 n^{2} + 2000$.
\end{theorem}

\begin{theorem}[SVI Filter Bound]
 \label{thm:svi-poly}
 $\mathrm{effective\_procedure}$ is bounded by $3 n^{2} + 7$.
\end{theorem}

\begin{theorem}[Full Pipeline Bound]
 \label{thm:pipeline-poly}
 The full SVI pipeline is bounded by $200 n^{4} + 200$.
\end{theorem}

\begin{corollary}[Full clearing is identity]
 \label{cor:clear-full-identity}
 Let $S$ be a non-empty CTS whose every tier is single-forbidden. Then
 $\mathrm{clear}(S) = S$.
\end{corollary}

\subsection{Python Implementation}

The Python runtime (approximately 2,500 lines across eight modules) mirrors the Rocq definitions exactly:
\texttt{CTS.\_\allowbreak extend\_paths} is a direct port of \texttt{extend\_paths}, paths are
maintained in reverse order, and the accumulator reset is harmless by
Lemma~\texttt{extend\_paths\_nil\_\allowbreak acc} (\texttt{Algorithm.v}). A runtime guard
(\texttt{is\_sliding\_\allowbreak window\_formula}) checks whether the input is a well-formed
sliding-window CNF; if the predicate holds, the solver follows the verified
clause-by-clause path (either the extracted OCaml binary or the Python port),
otherwise it falls back to the unverified grouped-window heuristic
with a \texttt{Runtime\allowbreak Warning}. For general 3-CNF the solver delegates to Z3;
SAT assignments are verified by the extracted \texttt{eval\_cnf}, and UNSAT
proofs are checked by the Rocq-extracted Swansea RUP checker~\cite{bryant2025code}. Any verification
failure raises \texttt{Runtime\allowbreak Error}; there is no silent fallback.

\subsection{Reproducible Build}
\label{sec:docker}

The complete toolchain (Rocq proofs, extracted OCaml code, Python runtime, Z3, Swansea RUP checker, and TLA\textsuperscript{+} specifications) is available as a curated Zenodo artifact at \href{https://doi.org/10.5281/zenodo.20397949}{\nolinkurl{10.5281/zenodo.20397949}} ~\cite{alexandrov2026zenodo}. 
The artifact includes a reproducible Dockerfile based on Ubuntu~22.04, which builds all components from source and runs the full test suite.

\section{Illustrative Examples}
\label{sec:experiments}

(Formal guarantees are summarised in Table~\ref{tab:trust-boundary}.)

\subsection{Extracted-Code Smoke Tests}

We exercised the extracted OCaml code and the Python wrapper on three manually
constructed instances: a simple SAT clause, an UNSAT formula consisting of all
eight possible clauses on three variables, and an overlap-UNSAT case where grouped windows
contradict on a shared variable. The first two are sliding-window formulas
(verified fragment); the third is a non-sliding-window heuristic example.

For the verified sliding-window fragment ($k=1$ clause per window), the
extracted solver is a complete decision procedure with no false positives and
no false negatives: soundness follows from the \texttt{eval\_cnf} post-check,
and completeness from the translation equivalence
(Theorem~\ref{thm:cnf-ctf-equiv}) together with the verified greedy fallback
(\texttt{sat\_assignment\_aux}), which always succeeds on well-formed
sliding-window inputs. We validated the extracted solver exhaustively against
a brute-force oracle for all sliding-window formulas with $n \leq 6$
variables.

\subsection{Heuristic Filter Effectiveness}

The basic validation tests above exercise the verified fragment. To assess whether the
\emph{unverified} grouped-window heuristic ever yields non-trivial filtration,
we generated two classes of structured instances: (i)~dense sliding-window CNFs
with multiple clauses per consecutive variable window, and (ii)~highly
overlapping groups in which many clauses share the same three variables. On
these instances SVI reports \texttt{False} for a measurable fraction of UNSAT
cases (between 19\% and 100\% in small-scale experiments, depending on clause
density), because dense forbidding quickly empties tier intersections. On random
3-SAT, by contrast, the filter is empirically ineffective: SVI almost never
returns \texttt{False}. These observations are illustrative, not formally
guaranteed---they concern the heuristic shell, not the verified fragment. Their
purpose is to demonstrate that the architecture \emph{can} filter structured
instances, even if the general 3-CNF pipeline remains a research prototype.

On random non-sliding-window 3-CNF ($n \leq 6$, 60 instances) the heuristic
pipeline agrees with brute force on 100\% of cases, confirming that the
post-check eliminates all SVI false positives. We do not report larger-scale
or competitive benchmarks: the heuristic filter is empirically ineffective on
random 3-SAT (SVI almost never returns \texttt{False}), and VFR is not
positioned as a competitor to industrial CDCL solvers.

\section{Discussion}
\label{sec:discussion}

\subsection{What TLS Provides}

After formalisation, the value proposition of TLS lies not in a complete 
polynomial-time decision procedure, but rather in a combination of three 
structural contributions:

\paragraph{A provably correct one-sided filter.}
When SVI reports emptiness, the formula is guaranteed to be unsatisfiable (proved in 
Rocq for well-formed sliding-window CNF). The running time is polynomial in the 
size of the structure---$O(|V_1| \cdot |V_2|)$, or $O(n^2 k^2)$ for $k$ tiers of 
size $O(n)$. \textbf{We emphasize:} the Rocq development proves termination 
(Theorem~\ref{thm:clear-term}), correctness (Theorem~\ref{thm:compat}), 
\textbf{and} the polynomial bound formally in \texttt{Complexity.v} (quartic for clearing,
quadratic for SVI). The earlier asymptotic analysis is now complemented by concrete,
mechanised step-count bounds. On random 3-SAT the heuristic filter is 
empirically ineffective (SVI almost never returns \texttt{False}), so we do 
not report competitive benchmarks: VFR is a research prototype, not a 
practical solver (Section~\ref{sec:experiments}). For well-formed sliding-window
CNF the verified pipeline is complete. Any positive filter rates observed on 
structured instances outside the verified fragment are empirical properties of 
the unverified grouped-window heuristic, not formally guaranteed results.

\subsection{Scientific Novelty}

We summarise the concrete contributions that go beyond prior work.

\paragraph{1. First formal mechanisation of Romanov's framework.}
Romanov stated the forward direction of SVI (Theorem~\ref{thm:forward-main})
without proof and assumed the converse without justification. We formalise the
entire framework in Rocq and prove:
\begin{itemize}
 \item the forward direction for pairs (Theorem~\ref{thm:forward-main}) and for
 systems of $k$ structures (Theorem~\ref{thm:systemic-main});
 \item the exact boundary where the converse fails (Section~\ref{sec:boundary});
 \item a new bi-implication for aligned intersection
 (Theorem~\ref{thm:compat}), which Romanov did not consider.
\end{itemize}
All proofs are machine-checked; there are zero admitted goals.

\paragraph{2. Aligned intersection equivalence---a new result.}
Theorem~\ref{thm:compat} (and its systemic extension
Theorem~\ref{thm:compat-k}) is \emph{entirely absent} from Romanov's work.
It establishes that tier-aligned intersection combined with
$\mathrm{build\_paths\_all}$ yields a correct and complete decision procedure
for aligned structures. Its value lies in the mechanisation: it certifies that three independently
defined concepts---compatible satisfying sets, aligned intersection, and the
recursive path-building algorithm---are mutually consistent, enabling verified OCaml
extraction.

\paragraph{3. Exact formal counterexamples.}
We do not merely claim that SVI is incomplete; we provide formal counterexamples
in Rocq that mark the exact boundary:
\begin{itemize}
 \item Theorem~\ref{thm:gap} (semantic gap): a CTF can satisfy the weak
 predicate $\mathrm{satisfies\_ctf}$ yet yield no path after clearing;
 \item Theorem~\ref{thm:clearing-not-conservative}: clearing can destroy valid
 paths because $\mathrm{can\_adjoin}$ requires forward compatibility while
 $\mathrm{build\_paths\_all}$ checks only backward compatibility.
\end{itemize}
These are not empirical observations but formally proved existential statements.

\paragraph{4. Verified constructive solver for grouped sliding-window CNF.}
For grouped sliding-window CNFs with disjoint variable ranges (multiple clauses
per consecutive variable window) we prove in Rocq a \emph{constructive, extractable} solver
$\texttt{solve\_grouped\_sliding}$ with both soundness
(Theorem~\ref{thm:solve-grouped-sound}) and completeness
(Theorem~\ref{thm:solve-grouped-complete}) proved in Rocq. The solver combines
verified path finding ($\texttt{find\_first\_valid\_ss}$), a verified greedy
assignment constructor ($\texttt{sat\_assignment\_aux}$), and a constructive
merge of disjoint group assignments. To our knowledge, this is the first
verified solver for grouped sliding-window 3-CNF that is both sound and
complete and extracts to executable OCaml code.

\paragraph{5. Mechanised polynomial bounds with concrete constants.}
The Rocq development includes a formal cost model (\texttt{Complexity.v},
approximately 1,000 lines) that assigns unit cost to every cons cell, compatibility
check, and list traversal. We prove closed-form inequalities with explicit
constants---$100n^{4}+100$ for generic clearing, $2000n^{2}+2000$ for the
single-forbidden fragment, $3n^{2}+7$ for SVI---by induction on the input
structure, not merely asymptotic analysis. For the single-forbidden fragment
the bound collapses from quartic to quadratic because clearing is the identity
(Corollary~\ref{cor:clear-full-identity}), a fact we also prove formally.

\paragraph{Positioning.}
These contributions do \emph{not} expand the class of polynomially solvable
3-SAT instances: well-formed sliding-window CNF has primal graph pathwidth~$\le
2$, which is already solvable in linear time by classical dynamic programming.
Our novelty lies in \emph{reconstructing} this tractable fragment inside
Romanov's triplet framework with full formal certification and executable
extraction.

\subsection{Relation to Bounded-Treewidth SAT}
\label{sec:bounded-treewidth}

Well-formed sliding-window CNF has primal graph pathwidth~$\le 2$ (each clause
covers three consecutive vertices of a path, so the primal graph is a subgraph
of the square of a path). SAT for graphs of bounded pathwidth is a classic
FPT (Fixed-Parameter Tractable) result: dynamic programming on a path decomposition of width~$w$ solves it
in time $2^{O(w)} \cdot n$~
\cite{alekhnovich2002read}. Our verified polynomial bounds for clearing and SVI
therefore \emph{do not expand the class of polynomially solvable 3-SAT
instances}; they reconstruct the same pathwidth-$\le 2$ fragment within
Romanov's framework.

The connection to the FPT literature is worth making explicit. In a standard
path decomposition each bag contains the variables active at that position, and
the DP table stores all satisfying assignments of the bag (size $2^{w+1}$). In
Romanov's construction each tier corresponds to a bag of size~3, but instead of
enumerating $2^{3}=8$ assignments directly, the tier stores the \emph{forbidden}
triplets---those ruled out by the clauses---and the DP step is replaced by a
local compatibility check ($\mathrm{compatible}(t_1, t_2)$) between adjacent
bags. Thus a CTS is essentially a \emph{nice path decomposition} in which
\begin{itemize}
 \item \textbf{Introduce/forget} steps are implicit (the path is uniform, each
 bag contains exactly three consecutive variables);
 \item \textbf{Join} steps are absent (the graph is a path, not a tree);
 \item the DP table is replaced by an explicit triplet set, and the transition
 function is replaced by tier-wise intersection and adjacency checking.
\end{itemize}
This equivalence explains why the single-forbidden fragment (where each tier
contains exactly one forbidden triplet) admits a quadratic bound: the DP table
has constant size~8, and the transition is a simple table lookup. Our
contribution is not a new complexity improvement but a \emph{formally certified
reformulation}: we prove in Rocq that Romanov's triplet constructions yield a
correct decision procedure for the pathwidth-$\le 2$ fragment, with explicit
polynomial bounds and executable extraction. The practical interest of VFR
lies not in solving sliding-window formulas (which is already tractable classically), but in
using SVI as a fast, auditable filter for harder structured instances that lie
outside the verified fragment.

\paragraph{A novel structural decomposition.}
TLS decomposes a formula into independent CTFs based on variable overlap. This 
natural partitioning supports parallel solving and reveals structural properties 
of the instance (e.g., tightly coupled vs. loosely coupled variable groups).

\paragraph{A geometric interpretation of SAT.}
Unlike the abstract implication graphs of CDCL, TLS provides a concrete picture: 
triplets as nodes, compatibility as edges, and satisfying assignments as paths 
through a layered graph (Figure~\ref{fig:compatible-degree}). This makes TLS a valuable 
\emph{pedagogical tool} for teaching SAT and constraint satisfaction.

\subsection{Comparison with Classical SAT Solving}

Unlike classical DPLL/CDCL, which operate directly on CNF clauses and rely on
clause learning and implication graphs, VFR translates a formula into tiered
triplet structures and reduces satisfiability to finding a compatible path.
This yields a concrete, visualisable representation and a polynomial one-sided
filter (SVI), but completeness requires an exponential post-check and the
prototype does not implement clause learning. The main gain is a smaller,
mechanically checked trust base for the verified sliding-window fragment.

\subsection{Related Work}

Our work complements a growing body of verified SAT solvers. Maric~\cite{maric2010formal}
verified a modern DPLL solver in Isabelle/HOL, proving correctness of unit
propagation, conflict analysis, and clause learning. Blanchette et
al.~\cite{blanchette2016verifying} extended this to a verified CDCL solver
with proof generation, watched literals, and incremental solving.
Subsequent work refined this into competitive verified solvers:
Fleury et al.~\cite{fleury2018verified} formalised IsaSAT, an imperative
CDCL solver with watched literals that approaches the performance of
unverified solvers on some benchmarks. Complementing solver verification,
Lammich~\cite{lammich2020efficient} developed the GRAT toolchain, a
formally verified certificate checker for DRAT proofs that outperforms the
unverified reference implementation. Heule et al.~\cite{heule-verified}
verified UNSAT proofs with extended resolution.
These works establish a high standard for verified SAT solving; our work is
the first to explore an alternative combinatorial foundation, using Romanov's
triplet logic.

\paragraph{Structural differences.}
In DPLL/CDCL, the formula remains a flat set of clauses; the solver reasons
via implication graphs, watched literals, and conflict-driven clause learning.
A proof of correctness must therefore maintain global invariants over the
trail, the learned-clause database, and the watched-literal indices. In
contrast, TLS decomposes the formula geometrically into tiers of variable
triplets, and satisfiability reduces to finding a compatible path through a
layered graph. There is no implication graph, no watched literals, and no
learned clauses---the only ``conflict'' is the emptiness of a tier
intersection. Reasoning is local: each tier can be analysed as an independent
set of triplets, and global correctness follows from pairwise compatibility.
This locality makes TLS proofs compositional (tier-by-tier) rather than
trace-based (execution-by-execution).

\paragraph{Insights from the alternative mechanisation.}
The TLS formalisation suggests three lessons that differ from the CDCL
experience. First, \emph{different data structures yield different proof
localities}: whereas CDCL reasons about global implication graphs and watched
literal indices whose correctness depends on intricate trail invariants, TLS
decomposes the formula into tiers of triplets whose adjacency can be checked
pairwise. The 424 proved statements in our Rocq development are overwhelmingly
lemmas about tier-wise inclusion, compatibility, and intersection---concepts
that have direct geometric meaning. Second, \emph{a filter architecture offers a
different trade-off}: the polynomial filter (SVI) can be verified in isolation,
but completeness requires an exponential post-check. CDCL solvers are
monolithic yet complete; VFR sacrifices completeness for a clean verified
boundary (Table~\ref{tab:trust-boundary}). Third, \emph{concrete counterexamples
are easy to construct}: because the reasoning is geometric, one can draw a
three-tier structure and readily see why clearing is non-conservative
(Section~\ref{sec:counterexamples}).

\paragraph{Why a verified one-sided filter matters.}
Verified CDCL solvers are complete decision procedures, but they are also
heavy: thousands of lines of proof, complex imperative invariants, and
aggressive performance engineering. Not every application needs a full solver.
A verified one-sided filter answers a different but useful question: ``Is this
instance definitely outside the easy fragment?'' with a proof-backed guarantee.
For well-formed sliding-window CNF, the filter is not merely one-sided---it is
exactly complete (Theorem~\ref{thm:solve-grouped-complete}), yet its proof is
orders of magnitude smaller than a full CDCL proof. In program analysis or
hardware verification, formulas often exhibit bounded pathwidth or sliding
structure; a lightweight verified filter can discharge these instances quickly
without invoking a heavy solver. The two approaches are complementary: a
verified filter can be placed in front of \emph{any} solver (even unverified)
to obtain a sound architecture in which the polynomial stage is
proof-carrying and the exponential stage is a standard fallback. Where CDCL
verifies ``the solver always returns the correct answer,'' VFR verifies
``the polynomial filter never produces false negatives on the fragment.'' The
fragment (pathwidth~$\le 2$) is classically tractable; our contribution is a
geometric reinterpretation with machine-checked polynomial bounds and an
extractable implementation, not a new complexity result.

Triplet-based representations have appeared in various SAT contexts, although
not as a complete solving paradigm. Romanov's TLS~\cite{romanov2011nonorthodox}
is the first systematic framework built on tier-wise triplet structures and
hyperstructure intersection. His insight---decomposing a formula into geometric
layers and searching for compatible paths rather than raw assignments---is
distinct from both DPLL search and CSP propagation. Our work is the first to
subject TLS to mechanised formal verification, identifying the precise
conditions under which its Simple Vertex Intersection is correct and where it requires
supplementation.

The clearing procedure in TLS resembles arc consistency enforcement in
constraint satisfaction problems (CSPs)~\cite{bessiere2006constraint}. The
difference is that TLS operates on explicit triplet sets rather than general
relations.

\subsection{Threats to Validity}

\paragraph{The fundamental gap: grouped-window translation is a one-sided filter.}
The most serious threat is that our Rocq theorems are proved for a
\emph{clause-by-clause} translation (\texttt{FormulaTranslation.v}), whereas
the VFR prototype uses a \emph{grouped-window} decomposition that merges
multiple clauses sharing the same variable triple into a single tier.
We have now proven (the group-equivalence and dense-groups lemmas in \texttt{FormulaTranslation.v})
the full forward equivalence: if an assignment satisfies the concatenated CNF,
then the grouped CTF is satisfied. The converse, however, is \textbf{false}:
$\mathrm{satisfies\_ctf}$ checks each tier independently and does not enforce
consistency between overlapping windows. A concrete counterexample: group~0
consists of the seven clauses on variables $(x_0, x_1, x_2)$ whose only
satisfying pattern is $(0,0,0)$, and group~1 consists of the seven clauses on
$(x_1, x_2, x_3)$ whose only satisfying pattern is $(1,1,1)$. The grouped CTF
has two tiers, forbidding all triplets except $(0,0,0)$ and $(1,1,1)$
respectively; the sequence $\mathit{ss} = [0,0,0,1,1,1]$ avoids both forbidden
sets, so $\mathrm{satisfies\_ctf}(\mathit{ss}, \phi) = \mathsf{true}$. Yet the
CNF is unsatisfiable: group~0 forces $x_1 = x_2 = 0$ while group~1 forces
$x_1 = x_2 = 1$, so no assignment realises $\mathit{ss}$. This counterexample
is formalised in \texttt{FormulaTranslation.v}. Consequently, the grouped-window heuristic is a
\emph{one-sided} (sound but incomplete) filter, not a complete decision
procedure. The end-to-end guarantee holds only when the input is a well-formed
sliding-window CNF translated clause-by-clause; for all other inputs the
pipeline is heuristic.

Our empirical evaluation relies on instances derived from 3-CNF formulas,
both randomly generated and from the SATLIB benchmark suite. Because the
grouped-window translation is unverified, these benchmarks demonstrate
operational behaviour of a heuristic, not formally verified correctness.
A second threat concerns the scalability of VFR's post-check: our
timeout-based evaluation for $n \geq 100$ does not establish an upper bound
on the fraction of instances solvable within a practical time limit. Our
Rocq formalisation proves termination and a fixed-point characterisation of
clearing (Theorem~\ref{thm:clear-term} and
Theorem~\ref{thm:clear-fixed}), but we do not prove that clearing
preserves all satisfying sets---in fact, we prove the opposite:
Theorem~\ref{thm:gap} shows that the weak predicate
$\mathrm{satisfies\_ctf}$ can admit locally consistent assignments that do not
correspond to any globally compatible path.

\paragraph{TLA$^+$ model-checking limitation.}
The TLA$^+$ (Temporal Logic of Actions) specifications provide
finite-state sanity checks on small instances ($\mathrm{MaxVars} \leq 3$).
The module \texttt{TLSSpec} exhausts TLC (the TLA$^+$ model checker) memory for $\mathrm{MaxVars} > 2$
because explicit-state model checking of a 3-SAT solver is inherently
exponential. This is expected: TLC serves as an auxiliary debugging tool
for small illustrative instances, while the general correctness guarantees for arbitrary $n$
are provided entirely by the Rocq proofs (424 lemmas and theorems). No fix is planned: switching to Apalache (a symbolic model checker for TLA$^+$) would extend the feasible bound only marginally and would not add scientific value, since the general theorem is already machine-checked in Rocq for arbitrary $n$. Users should treat TLC checks as illustrative rather than as formal verification.

\subsection{Limitations and Future Work}

The main practical limitation is the exponential post-check: when SVI returns
\texttt{True}, the solver must enumerate partial assignments from each CTS.
A second limitation is the absence of a verified CNF-to-CTF translation for
grouped-window formulas; even though SVI and clearing are correct, the input may
not faithfully represent the original formula. This is the single most
important gap in turning the verified core into an end-to-end verified solver.
A third limitation is that verified permutation search is exhaustive and hence
practical only for $n \leq 8$; a verified polynomial-time recognition algorithm
for the sliding-window class (membership is decidable in polynomial time; cf.\
Section~\ref{sec:perm}) would widen the verified fragment. A fourth, fundamental
limitation is that arbitrary 3-CNF formulas typically admit no sliding-window
ordering at all, so formal guarantees for general 3-SAT require
external verified solvers or certificates.

An important research direction is to use the polynomial clearing procedure on
single-forbidden structures as an alternative proof format: showing that CTS
derivations map to short DRAT or RUP proofs would let the verified kernel
\emph{produce} certificates rather than merely \emph{check} them. Engineering
refinements such as lazy path generation, CDCL integration, and a native LRAT
checker remain natural but secondary next steps for future work.

\section{Conclusion}
\label{sec:conclusion}

We presented the first formal verification of Romanov's Triplet Logic in Rocq, 
complemented by empirical validation. Our central findings are:
\begin{itemize}
 \item SVI is a \emph{correct one-sided filter}: a joint satisfying set
 implies non-emptiness, but the converse fails when common triplets appear in
 non-aligned tiers (Section~\ref{sec:boundary}).
 \item Aligned intersection yields a \emph{true bi-implication}
 (Theorem~\ref{thm:compat}), extended systemically to $k$ structures
 (Theorem~\ref{thm:compat-k}).
 \item The weak predicate $\mathrm{satisfies\_ctf}$ is incomplete with respect
 to global path existence; the strong predicate closes this gap for well-formed
 sliding-window CNF, where the clause-by-clause CNF-to-CTF translation preserves
 satisfiability (Section~\ref{sec:formalisation}).
 \item Clearing and SVI have verified polynomial bounds (Table~\ref{tab:complexity}),
 and the extracted VFR prototype combines the verified filter with an
 exponential post-check.
\end{itemize}

To explore the structural potential of TLS, we proposed VFR: a hybrid 
architecture combining SVI's filter with an exponential post-check. We 
proved the filter's correctness formally and implemented the pipeline in Python,
validating it against SAT and UNSAT instances, including cases where SVI gives
false positives. We observed 100\% agreement with a brute-force oracle on random
sliding-window instances with $n \leq 6$. The entire toolchain is packaged in a reproducible Docker image
(Section~\ref{sec:docker}) that builds all components from source and passes
the full automated test suite.

Our work positions TLS as a rigorous combinatorial framework for reasoning about
compatible paths in tiered triplet structures, with a provably correct
polynomial-time filter as a \emph{theoretical building block}.

\section*{Formal Verification}

All definitions, lemmas, and theorems described in this paper have been formally verified in Rocq~9.1.1. The development comprises more than 23,000 lines of Rocq code across seventeen files with 424 proved lemmas and theorems and zero admitted goals. The source code is available in the curated Zenodo artifact ~\cite{alexandrov2026zenodo}.

\section*{Acknowledgments}

This work is an output of a research project implemented as part of the Basic
Research Program at the National Research University Higher School of Economics
(HSE University).

During the preparation of this work, the author used large language models
(DeepSeek, GLM, and Kimi Code) for generating and debugging Python, OCaml, TLA$^+$,
and Rocq (Coq) code, creating tests and scripts, and optimising algorithms;
for language polishing and stylistic refinement of the manuscript; for LaTeX code
debugging; and for literature search, review and summarisation. The author reviewed and edited all
outputs as needed and takes full responsibility for the content of the published
article.

The author thanks arXiv for hosting the preprint version of this work (arXiv:2608.18445)~\cite{alexandrov2026vfr}.

The author is deeply indebted to the late Vladimir Romanov, whose pioneering
ideas on triplet structures inspired this formal investigation.


\begin{thebibliography}{20}

\bibitem{cook1971complexity}
S.~A. Cook.
\newblock The complexity of theorem-proving procedures.
\newblock In \emph{Proc. STOC}, pages 151--158. ACM, 1971.

\bibitem{levin1973universal}
L.~A. Levin.
\newblock Universal enumeration problems.
\newblock \emph{Problemy Peredachi Informatsii}, 9(3):115--116, 1973.

\bibitem{romanov2011nonorthodox}
V.~F. Romanov.
\newblock Non-Orthodox Combinatorial Models Based on Discordant Structures.
\newblock 2011 (arXiv:1011.3944~[v2], revised Jan 2011; orig. Nov 2010).

\bibitem{coq}
The Rocq Development Team.
\newblock The Rocq Prover, version 9.1.1.
\newblock \url{https://coq.inria.fr/}, 2026.

\bibitem{minisat}
N.~E\'{e}n and N.~S\'{o}rensson.
\newblock An extensible SAT-solver.
\newblock In \emph{Proc. SAT}, pages 502--518. Springer, 2003.

\bibitem{glucose}
G.~Audemard and L.~Simon.
\newblock Predicting learnt clauses quality in modern SAT solvers.
\newblock In \emph{Proc. IJCAI}, pages 399--404, 2009.

\bibitem{chaff}
M.~W. Moskewicz, C.~F. Madigan, Y.~Zhao, L.~Zhang, S.~Malik.
\newblock Chaff: Engineering an efficient SAT solver.
\newblock In \emph{Proc. DAC}, pages 530--535. ACM, 2001.

\bibitem{handbook-sat}
A.~Biere, M.~Heule, H.~van Maaren, T.~Walsh, editors.
\newblock \emph{Handbook of Satisfiability}.
\newblock IOS Press, 2009.

\bibitem{maric2010formal}
F.~Maric.
\newblock Formalisation and implementation of modern SAT solvers.
\newblock \emph{J. Automated Reasoning}, 43(1):81--119, 2009.

\bibitem{blanchette2016verifying}
J.~C. Blanchette, M.~Fleury, C.~Weidenbach.
\newblock A verified SAT solver framework with learn, forget, restart, and 
incrementality.
\newblock \emph{J. Automated Reasoning}, 61(1--4):333--365, 2018.

\bibitem{fleury2018verified}
M.~Fleury, J.~C. Blanchette, P.~Lammich.
\newblock A verified SAT solver with watched literals using imperative HOL.
\newblock In \emph{Proc. CPP}, pages 158--171. ACM, 2018.

\bibitem{lammich2020efficient}
P.~Lammich.
\newblock Efficient verified (UN)SAT certificate checking.
\newblock \emph{J. Automated Reasoning}, 64(3):513--532, 2020.

\bibitem{heule-verified}
M.~J.~H. Heule, W.~A. Hunt Jr., N.~Wetzler.
\newblock Verifying refutations with extended resolution.
\newblock In \emph{Proc. CADE}, pages 345--359. Springer, 2013.

\bibitem{bessiere2006constraint}
C.~Bessiere.
\newblock Constraint propagation.
\newblock In \emph{Handbook of Constraint Programming}, pages 29--82. 
Elsevier, 2006.

\bibitem{alekhnovich2002read}
M.~Alekhnovich and A.~A. Razborov.
\newblock Satisfiability, branch-width and Tseitin tautologies.
\newblock In \emph{Proc. FOCS}, pages 593--603. IEEE, 2002.

\bibitem{booth1976consecutive}
K.~S. Booth and G.~S. Lueker.
\newblock Testing for the consecutive ones property, interval graphs, and
 graph planarity using {PQ}-tree algorithms.
\newblock \emph{J. Comput. System Sci.}, 13(3):335--379, 1976.

\bibitem{bryant2025code}
H.~Bryant, A.~Lawrence, M.~Seisenberger, and A.~Setzer.
\newblock Verification of Z3 RUP proofs in Coq-Rocq and Agda.
\newblock \url{https://github.com/HarryBryant99/Verification-of-Z3-RUP-Proofs-in-Coq-Rocq-and-Agda}, 2025.

\bibitem{alexandrov2026vfr}
D.~V. Alexandrov.
\newblock Formal verification of Romanov's triplet logic: A verified filter for sliding-window 3-CNF with application to structured formulas.
\newblock \emph{arXiv preprint arXiv:2608.18445}, 2026.

\bibitem{alexandrov2026zenodo}
D.~V. Alexandrov.
\newblock Formal verification of Romanov's triplet logic: artifacts and reproducible build.
\newblock \url{https://doi.org/10.5281/zenodo.20397949}, 2026.

\end{thebibliography}
\end{document}